\documentclass[12pt,a4paper]{article}

\usepackage{amssymb}
\usepackage{amsfonts}
\usepackage{amsmath}
\usepackage{amsthm}
\usepackage[all]{xy}
\usepackage{latexsym}
\usepackage{enumerate}

\DeclareMathOperator{\rank}{rank}

\newcommand{\bbN}{\mathbb{N}}
\newcommand{\bbNp}{\mathbb{N}^*}

\newcommand{\ef}{\mathbb{F}}

\newcommand{\efq}{\ef_q}
\newcommand{\eftwo}{\ef_2}

\newcommand{\bs}{\boldsymbol}

\newcommand{\veps}{\varepsilon}
\newcommand{\tl}{\tilde}

\newcommand{\dham}{\mathrm{d}_{\mathrm{H}}}

\newcommand{\lm}{\textsc{lm}}

\newcommand{\supp}{\mathrm{supp}}

\newcommand{\synd}[1]{S^{(\bs{#1})}}
\newcommand{\tlsig}{\tl{\sigma}}
\newcommand{\tlom}{\tl{\omega}}
\newcommand{\tlbsy}{\tl{\bs{y}}}
\newcommand{\syndtly}{S^{(\tlbsy)}}

\newcommand{\plus}[1]{{#1}^{+}}
\newcommand{\Gplus}{\plus{G}}
\newcommand{\gplus}{\plus{g}}
\newcommand{\bgplus}{\plus{\bs{g}}}
\newcommand{\Mplus}{\plus{M}}
\newcommand{\bder}{\mathbf{der}}
\newcommand{\broot}{\mathbf{root}}

\newcommand{\rmax}{r_{\max}}

\newcommand{\kavcic}{Kav\v ci\' c} 
\newcommand{\pade}{Pad\'e}

\newcommand{\bh}{\bs{h}}

\newcommand{\lmmo}{\lm_{-1}}

\newcommand{\fplus}{\plus{f}}
\newcommand{\bfplus}{\plus{\bs{f}}}

\newcommand{\groebner}{Gr\"obner}
\newcommand{\koetter}{K\"otter}

\newtheorem{proposition}{Proposition}[section]
\newtheorem{definition}[proposition]{Definition}
\newtheorem{corollary}[proposition]{Corollary}

\newtheorem{theorem}[proposition]{Theorem}
\newtheorem{remark}[proposition]{Remark}

\newtheorem*{remarknn}{Remark}
\newtheorem*{lemmann}{Lemma}

\title{A \groebner-bases approach to syndrome-based fast Chase
decoding of Reed--Solomon codes }
\date{\today} 
\author{Yaron Shany and Amit Berman\thanks{The
authors are with Samsung  Semiconductor Israel R\&D Center, 146 Derech
Begin St., 6492103, Tel Aviv, Israel. Emails: \{yaron.shany,
amit.berman\}@samsung.com}}   

\begin{document}
\maketitle

\begin{abstract}
We present a simple syndrome-based fast Chase decoding algorithm for
Reed--Solomon (RS) codes. Such an algorithm was initially presented by
Wu (IEEE Trans.~IT, Jan.~2012), building on 
properties of the Berlekamp--Massey (BM) algorithm. Wu devised a fast
polynomial-update algorithm to construct the error-locator polynomial
(ELP) as the solution of a certain linear-feedback shift register
(LFSR) synthesis problem. This results in a conceptually complicated
algorithm, divided into $8$ subtly different cases. Moreover, Wu's
polynomial-update algorithm is not immediately suitable for working
with vectors of evaluations. Therefore, complicated modifications were
required in order to achieve a true ``one-pass'' Chase decoding
algorithm, that is, a Chase decoding algorithm requiring $O(n)$
operations per modified coordinate, where $n$ is the RS code length. 

The main result of the current paper is a conceptually simple
syndrome-based fast Chase decoding of RS codes. Instead of developing
a theory from scratch, we use the well-established theory of
\groebner\ bases for modules over $\efq[X]$ (where $\efq$ is the
finite field of $q$ elements, for $q$ a prime power). The basic
observation is that instead of Wu's LFSR synthesis problem, it is much
simpler to consider ``the right'' minimization problem over a
\emph{module}. The solution to this minimization problem is a simple 
polynomial-update algorithm that avoids syndrome updates and works
seamlessly with vectors of evaluations. As a result, we obtain a 
conceptually simple algorithm for one-pass Chase decoding of RS codes. 
Our algorithm is general enough to work with any algorithm
that finds a \groebner{} basis for the solution module of the key
equation as the initial algorithm (including the Euclidean algorithm),
and it is not tied only to the BM algorithm. 
\end{abstract}

\section{Introduction}

\subsection{Motivation and known results}
The subject of decoding Reed--Solomon (RS) codes beyond half the
minimum distance has been extensively studied over the years. The
breakthrough work of Guruswami and Sudan \cite{GS99} (following the original
work of Sudan \cite{S97}) presented interpolation-based {\it hard-decision (HD)}
list decoding of  RS codes up to the so-called Johnson radius in
polynomial time. Wu \cite{Wu08} presented an even more efficient
HD list decoding algorithm for decoding RS codes up to the Johnson
radius. \koetter{} and Vardy \cite{KV} extended the Guruswami--Sudan  
algorithm to take {\it channel reliability information} into account, thus presenting a
polynomial-time {\it soft-decision (SD)} decoding algorithm for RS codes. 

Before \cite{KV}, it seems reasonable to say that the main SD decoding
algorithms for block codes with an efficient HD decoder in general,
and for RS codes in particular, were the {\it generalized minimum distance
(GMD)} decoding of Forney \cite{F66}, and the Chase decoding
algorithms \cite{Chase}. GMD decoding consists of repeated applications of
errors-and-erasures decoding, while successively erasing an even
number of the least reliable coordinates. 

In Chase decoding, there is some pre-determined list of {\it test
error patterns} on the $\eta$ least reliable coordinates for some
small $\eta$ (typically, $\eta\leq \lfloor d/2\rfloor$, where $d$ is
the minimum Hamming distance of the code). For example, this list may consist
of all possible non-zero vectors, all vectors of a low enough weight,
a pre-defined number of random vectors, etc.. The decoder successively
runs on error patterns from the list. Each such error pattern is
subtracted from the received word, and the result is fed to an HD
decoder. If the HD decoder succeeds, then its output is saved into the
output list of the decoder. 

Informally, the list of test error patterns in Chase decoding 
should be a good covering code for all likely error patterns on the
least reliable coordinates (see \cite{NPN11}). At least heuristically,
this suggests that in order to achieve a substantial gain over HD
decoding, the number of test error patterns should grow exponentially
with $d$.  

Despite this exponential nature of Chase decoding, for high-rate
codes of moderate length, it is known to have a better
complexity/performance tradeoff than other algebraic SD decoding
algorithms, including the \koetter{}--Vardy algorithm (see,
e.g., \cite{Wu12}). For this reason, Chase decoding of RS codes is still of
great interest.
The idea behind fast Chase decoding algorithms is to share
computations between HD decodings of different test error
patterns.\footnote{This idea is at the heart of all
fast Chase decoding algorithms, including \cite{BK10}, \cite{Wu12},
and \cite{XCB20}.} For example, if a new error pattern differs from the 
previous one in one additional non-zero coordinate, it seems plausible
that there is no need to run a full HD decoding algorithm for the new
error pattern, and that intermediate results from the previous HD
decoding can be used in order to reduce the complexity of the new HD
decoding.  

It is well-known that HD decoding of RS codes has
complexity $O(dn)$ (where $n$ is the length 
of the RS code), and that this complexity is
governed by the exhaustive root search, rather than by the algorithm
for finding the {\it error-locator polynomial (ELP)} (such as the
Berlekamp--Massey (BM) algorithm), which has a complexity of
$O(d^2)$. 

In \cite{Wu12}, Wu defines a {\it one-pass} Chase decoding algorithm
as a Chase algorithm that has 
the following properties: (1) For any test error pattern $\bs{z}$
of (Hamming) weight $w$, there is some sequence
$\bs{z}_1,\ldots,\bs{z}_{w-1}$ in the list of test error
patterns, such that for all $i$, the weight of $\bs{z}_i$
equals $i$, and such that $\supp(\bs{z}_1)\subset \cdots \subset
\supp(\bs{z}_{w-1})\subset \supp(\bs{z})$,\footnote{For
$\bs{b}:=(b_1,\ldots,b_n)$, $\supp(\bs{b}):=\{i|b_i\neq 0\}$ is the
{\it support} of $\bs{b}$.} (2) The algorithm produces
decoding results for all the sequence
$\bs{z}_1,\ldots,\bs{z}_{w-1},\bs{z}$ in a complexity of $O(wn)$
finite field operations.  In 
particular, if $w=O(d)$, then the complexity for decoding the subset
$\{\bs{z}_1,\ldots,\bs{z}_{w-1}, \bs{z}\}$ is $O(dn)$, just like HD
decoding. Put differently, the complexity is $O(n)$ per each additional
modified coordinate. Note that in a na\"ive application of Chase
decoding, the complexity of decoding the above sequence is $O(dn)$ per
each additional modified coordinate.

Before \cite{Wu12}, there have been several one-pass and ``almost''
one-pass Chase decoding algorithms for BCH and RS codes, where by
``almost'' we mean that some of these algorithms satisfied the above
complexity requirement only 
for producing the ELP, but not for the essential following exhaustive
root search. These algorithms include the low-complexity
interpolation-based algorithm of Bellorado and \kavcic\ \cite{BK10}
for RS codes (based on the Guruswami--Sudan algorithm), and the
algorithm of Kamiya \cite{K01} for binary BCH codes, based on the
Welch--Berlekamp algorithm. Also, in the context of \cite{BK10}, Zhu
{\it et al.} \cite{ZZW09} introduced an efficient method for backward
interpolation, which enables to cancel the interpolation in one
point. This allows \cite{ZZW09} to order the $2^{\eta}$ test
vectors according to adjacent vertices of the binary hypercube (Gray
code), thus avoiding the need to save 
$2^{\eta-1}$ intermediate interpolation results on the decoding tree
of \cite{BK10}.  For a thorough literature review on fast
Chase decoding algorithms before \cite{Wu12}, we refer to
\cite{Wu12}. 
 
Focusing on RS codes and considering the algorithm of \cite{BK10},
we note that this algorithm works in the ``time domain,'' i.e., on the
received vector itself, rather than on the {\it syndrome}. As noted in
\cite[p.~946]{BK10}, in the context of fast Chase 
decoding, it is somewhat easier to work directly on the received
vector rather than on the syndrome, because the syndromes
of similar test error patterns are very far from similar.  

For decoding high-rate codes, it is typically beneficial to replace
the long received vector by the short syndrome once 
and for all before the decoding begins. In his important paper
\cite{Wu12}, Wu introduced a true one-pass Chase decoding algorithm based on
the BM algorithm. Thus, Wu introduced a solution both for the problem
of handling the exhaustive root searches while maintaining a
complexity of $O(n)$ per modified coordinate, and for the need for a
syndrome-based algorithm. 

After Wu's work, \cite{ZZW11} proposed a backward step for Wu's
algorithm for binary BCH codes. Additional
time-domain Chase decoding algorithms for binary BCH codes were
developed, e.g., in \cite{Z13}. Also, for RS codes, time-domain Chase
decoding algorithms based on basis reduction for univariate polynomial
modules were presented in \cite{XCB20} and the references
therein. Inspired by \cite{NZ03} and \cite{XCB19}, the fast Chase
algorithm of \cite{XCB20} decreases the average complexity and latency
over that of \cite{BK10}, while maintaining the worst-case
complexity. It should be noted that for the setup of \cite{BK10}, the
worst-case complexity of \cite{BK10} and \cite{XCB20} is
$O(2^{\eta}\cdot(n-k)^2)$, which is $O(2^{\eta}\cdot n^2)$ if the
asymptotic rate is $<1$, while that of a true one-pass Chase decoding
algorithm is $O(2^{\eta}\cdot n)$, as it requires $O(n)$ operations
per edge of the decoding tree, for a maximum of $2^{\eta}-1$
visited edges.  

\subsection{Our results}
We use the well-established tool of \groebner{} bases for modules over
$\efq[X]$ to derive an algorithm for syndrome-based fast Chase
decoding of RS codes. The main observation is that instead of Wu's
LFSR synthesis problem, it is much simpler to consider ``the right''
minimization problem over a module. This minimization problem can be 
solved by adopting \koetter{}'s \groebner{} basis algorithm, in the general
form appearing in \cite[Sect.~VII.C]{McE03}.

\begin{itemize}

\item We present a clean and simple polynomial-update algorithm for
fast Chase decoding, namely, Algorithm A of Section
\ref{sec:basic}. This algorithm is considerably simpler than 
Algorithm 1 of \cite{Wu12}, which is divided into $8$ intricately
different cases. Besides of the obvious benefit of having a clear and
short algorithm and the theoretical interest of finding further
connections between decoding algorithms and \groebner{} bases, there is
also a practical benefit in a simply-presented algorithm, being easier
to implement and debug. 

\item As opposed to Algorithm 1 of \cite{Wu12}, our polynomial-update 
algorithm (Algorithm A) is automatically 
suited for working with vectors of evaluations, and it is easily
converted into Algorithm B, which has the required $O(n)$ complexity
per modified coordinate. Again, Algorithm B is considerably cleaner
and simpler than Algorithm 2 of \cite{Wu12}, which, besides of being
long and including $8$ different cases, requires the introduction of
auxiliary polynomials without a clear meaning.

\item As opposed to the algorithms of \cite{Wu12}, Algorithms A and B
of the current paper are not tied to the BM algorithm as the initial
HD decoding algorithm, and can practically work with any of the
existing syndrome-based HD decoding algorithms. In some detail,
Algorithms A and B can be initiated with any algorithm that finds a
\groebner{} basis for the solution module of the key equation (for an
appropriate monomial ordering). As shown by Fitzpatrick \cite{Fitz95},
practically any of the existing syndrome-based HD decoding algorithms
can be put in this form, including the Euclidean algorithm.

\item On the practical side, we present Algorithm C, which is a
variant of Algorithm A that runs on low-degree polynomials and has a
lower complexity than Algorithm 1 of \cite{Wu12}.

\end{itemize}

\subsection{Organization}
Section \ref{sec:prelim} includes the notation used throughout the
paper, some basic definitions, and a review of required 
known results on algebraic decoding of (generalized) RS codes. Wu's
idea of fast Chase decoding on a tree is also recalled in this
section.  

The new minimization problem over an $\efq[X]$-module and its relation
to fast Chase decoding are presented in Section \ref{sec:module},
which is the heart of the paper. The minimization problem is
translated into an 
application of \koetter{}'s \groebner{} basis algorithm in Section
\ref{sec:algs}. The polynomial update algorithm is presented in
Subsection \ref{sec:basic}, and the true one-pass Chase decoding
algorithm, working with vectors of evaluations, is presented in
Subsection \ref{sec:eval}. Section \ref{sec:algs} is concluded by
Subsection \ref{sec:high}, which presents an overall high-level
description of the entire decoding process. Finally, Section
\ref{sec:conclusion} includes some conclusions and open questions. 

The paper includes two appendices, containing some interesting
supplemental results. In Appendix \ref{app:indirect}, which may be
considered as the counterpart of \cite[Lemma 5 (ii)]{Wu12}, we
consider a certain interesting case that is not required for the
algorithms of Section \ref{sec:algs}, and show that even in this case, 
the ELP can be extracted from the output of the polynomial-update
algorithm. Appendix \ref{app:simplifications} includes some practical
simplifications of Algorithm A: a method for avoiding the need to work
with two \emph{pairs} of polynomials, so that it is possible to work
with just two scalar polynomials, a heuristic stopping
condition for (almost) avoiding unnecessary exhaustive root searches,
and a method that uses a transformation that significantly reduces the
degrees of the updated polynomials, and results in the low-complexity
Algorithm C.

\section{Preliminaries}\label{sec:prelim}
\subsection{Generalized Reed--Solomon codes}\label{sec:rs}
Let $q$ be a prime power, and let $\efq$ be the
finite field of $q$ elements. We will consider a primitive generalized
Reed--Solomon (GRS) code,\footnote{Since the most general GRS code (e.g.,
\cite[Sec.~5.1]{Roth}) may be obtained by shortening a primitive GRS
code, there is no loss of generality in considering only primitive GRS
codes.} $C$, of length $n:=q-1$ and minimum Hamming distance $d\in\bbNp$,
$d\geq 2$. In detail, let
$\tl{\bs{a}}=(\tl{a}_0,\ldots,\tl{a}_{n-1})\in(\efq^*)^n$ be a vector
of non-zero elements (where $\efq^*:=\efq\smallsetminus\{0\}$). For a
vector $\bs{f}=(f_0,f_1,\ldots,f_{n-1})\in 
\efq^n$, let $f(X):=f_0+f_1X+\cdots + f_{n-1}X^{n-1}\in
\efq[X]$. Now $C\subseteq \efq^n$ is defined as the set of all vectors
$\bs{f}\in\efq^n$ for which $\tl{a}(X)\odot f(X)$ has roots
$1,\lambda,\ldots,\lambda^{d-2}$  for some fixed primitive  $\lambda\in
\efq$, where $(-\odot-)$ stands for coefficient-wise multiplication
of polynomials.\footnote{For $f(X)=\sum_{i=0}^r f_i X^i$ and
$g(X)=\sum_{i=0}^s g_i X^i$, let $m:=\min\{r,s\}$, and define
$f(X)\odot g(X):=\sum_{i=0}^m f_ig_i X^i$.} We note that when
$(\tl{a}_0,\ldots,\tl{a}_{n-1})=(1,\ldots,1)$, $C$ is a Reed--Solomon
code. 

To recall the key equation \cite[Sec.~6.3]{Roth}, suppose that a
codeword $\bs{x}\in C$ is transmitted, and the  
received word is $\bs{y}:=\bs{x}+\bs{e}$ for some error vector
$\bs{e}\in \efq^n$.  For $j\in\{0,\ldots,d-2\}$, let
$S_j=\synd{y}_j:=(\tl{a}\odot y)(\lambda^j)$. The {\bf syndrome polynomial}
associated with $\bs{y}$ is
$\synd{y}(X):=S_0+S_1X+\cdots+S_{d-2}X^{d-2}$. By the definition of
the GRS code, the same syndrome polynomial is associated 
with $\bs{e}$. 

If $\bs{v}\in \efq^n$ is such that $v(X)=X^i$ for some $i\in
\{0,\ldots,n-1\}$, then $\synd{v}_j=(\tl{a}\odot
v)(\lambda^j)=\tl{a}_i(\lambda^{i})^j$, so  
that
\begin{equation}\label{eq:oneloc}
\synd{v}(X)=\tl{a}_i\big(1+\lambda^iX + \cdots +
(\lambda^i)^{d-2}X^{d-2}\big)\equiv 
\frac{\tl{a}_i}{1-\lambda^iX} \mod (X^{d-1}).
\end{equation}
So, if the error locators are some distinct elements
$\alpha_1,\ldots,\alpha_{\veps}\in \efq^*$ (where 
$\veps\in\{1,\ldots,n\}$ is the number of errors) and the
corresponding error values are
$\beta_1,\ldots,\beta_{\veps}\in\efq^*$, then  
\begin{equation}\label{eq:synd}
\synd{y}(X)=\synd{e}(X) \equiv \sum_{i=1}^{\veps}
\frac{\beta_i a_i}{1-\alpha_iX}\mod (X^{d-1}),
\end{equation}
where $a_i:=\tl{a}_{i'}$ for the $i'\in\{0,\ldots,n-1\}$ with
$\alpha_i=\lambda^{i'}$. 

Defining the {\bf error-locator polynomial (ELP)}, $\sigma(X)\in\efq[X]$, by
$$
\sigma(X):=\prod_{i=1}^{\veps}(1-\alpha_iX), 
$$
and the {\bf error-evaluator polynomial (EEP)}, $\omega(X)\in\efq[X]$,
by
$$
\omega(X):=\sum_{i=1}^{\veps} \beta_i a_i\prod_{j\neq
i}(1-\alpha_jX),
$$
it follows from (\ref{eq:synd}) that 
\begin{equation}\label{eq:key}
\omega\equiv \synd{y}\sigma \mod (X^{d-1}).
\end{equation}
Equation (\ref{eq:key}) is the so-called {\bf key equation}. 

Another useful relation is {\bf Forney's formula} (see, e.g.,
\cite[Sec.~6.5]{Roth}), which states that for all
$i\in\{1,\ldots,\veps\}$, \begin{equation}\label{eq:forney}
\beta_i a_i \sigma'(\alpha_i^{-1})=-\alpha_i w(\alpha_i^{-1}),
\end{equation}
where for a polynomial $f(X)$, $f'(X)$ stands for its formal
derivative. 

Let
$$
M_0 = M_0(\synd{y}) : =\big\{(u,v)\in \efq[X]^2\big|u\equiv
\synd{y}v \mod (X^{d-1})\big\}
$$
be the solution module of the key equation.\footnote{The reason for
the subscript ``$0$'' in $M_0$ will become apparent later, when we define
modules $M_r$ for each $r$ in Definition \ref{def:mr}.} Next,
we would like to 
recall that if the number of errors in $\bs{y}$ is up to $t:=\lfloor
(d-1)/2\rfloor$, then $(\omega,\sigma)$ is a minimal element in $M_0$
for an appropriate monomial ordering on $\efq[X]^2$

For background on monomial orderings and \groebner{} bases for
modules, see, e.g., \cite[Sec.~5.2]{CLO2} for the general case, and
\cite{Fitz95} for the special case of submodules of $K[X]^2$ (for $K$
a field), which is mostly sufficient for the current paper. Recall
that for $\ell\in \bbN$, a {\it monomial} in $K[X]^{\ell+1}$ is a vector of the
form $\bs{m}:=X^i\cdot \bs{u}_j$ for some $i\in \bbN$, and some
$j\in\{0,\ldots,\ell\}$, where $\bs{u}_j=(0,\ldots,0,1,0,\ldots,0)$,
and where the $1$ sits in the $j$-th position (counting from
$0$).\footnote{The reason for labeling coordinates with $0,1,\ldots$
rather than with $1,2,\ldots$ is that in some list-decoding
applications, it is convenient to identify $K[X]^{\ell+1}$ with the
polynomials in $K[X,Y]$ with $Y$-degree at most $\ell$, by mapping
$(f_0(X),\ldots,f_{\ell}(X))$ to $\sum_{j=0}^{\ell}f_j(X)Y^j$.} In 
such a case, we will say that $\bs{m}$ {\it contains} the $j$-th unit
vector.  

The monomial
ordering of the following definition is the special case of the
ordering $<_r$ of \cite{Fitz95} corresponding to $r=-1$. If a pair
$(f(X),g(X))$ is regarded as the bivariate polynomial $f(X)+Yg(X)$,
then this ordering is also the $(1,-1)$-weighted-lex ordering with
$Y>X$.   

\begin{definition}
{\rm
Define the following monomial ordering, $<$, on $\efq[X]^2$:
$(X^i,0)<(X^j,0)$ iff $i<j$, $(0,X^i)<(0,X^j)$ iff $i<j$, while
$(X^i,0)<(0,X^j)$ iff $i\leq j-1$.
}
\end{definition}

Unless noted otherwise, 
$\lm(u,v)$ will stand for the leading monomial of $(u,v)$ with
respect to the above monomial ordering, $<$. Also, a ``\groebner{}
basis'' will stand for a \groebner{} basis with respect to $<$. Finally,
$\dham(\cdot,\cdot)$ will stand for the Hamming distance. 

The following proposition is a special case of
\cite[Thm.~3.2]{Fitz95}. We include its  
simple and standard proof for completeness. 

\begin{proposition}\label{prop:rzero}
Using the above notation, suppose that $\dham(\bs{y},\bs{x})\leq
t$. Let $(u,v)\in M_0(\synd{y})\smallsetminus\{(0,0)\}$ satisfy
$\lm(u,v)\leq 
\lm(\omega,\sigma)$. Then there exists some $c\in \efq^*$ such 
that $(u,v)=c\cdot (\omega,\sigma)$. Hence, $(\omega,\sigma)$ is the
unique minimal element $(u,v)$ in $M_0$ with $v(0)=1$.
\end{proposition}

\begin{proof}
First, we claim that if there exist $(\tl{u},\tl{v}),(u,v)\in
M_0(\synd{y})$ and $d_1,d_2\in \bbN$ with $d_1+d_2<d-1$,
$\gcd(\tl{u},\tl{v})=1$,  $\deg(u),\deg(\tl{u})\leq d_1$, and
$\deg(v),\deg(\tl{v})\leq d_2$, then there exists a polynomial $f\in
\efq[X]$ such that $(u,v)=f\cdot(\tl{u},\tl{v})$. To see this, note
that from $u\equiv \synd{\bs{y}}v\mod (X^{d-1})$ and
$\tl{u}\equiv \synd{\bs{y}}\tl{v}\mod (X^{d-1})$, we get
$u\tl{v}\equiv \tl{u}v\mod (X^{d-1})$. In view of the above
degree constraints, the last congruence implies $u\tl{v}=
\tl{u}v$. Since $\gcd(\tl{u},\tl{v})=1$, we must have $\tl{u}|u$, 
$\tl{v}|v$, and $u/\tl{u}=v/\tl{v}$. This establishes the claim.

Now let $(u,v)\in M_0(\synd{y})$, and note that
$\gcd(\omega,\sigma)=1$.  If 
$\deg(v)>t\geq \deg(\sigma)$, then clearly
$\lm(u,v)>\lm(\omega,\sigma)=(0,X^{\deg(\sigma)})$. Similarly, if
$\deg(u)>t-1\geq \deg(\sigma)-1$, then
$\lm(u,v)>\lm(\omega,\sigma)$. Hence, we may assume w.l.o.g.~that
$\deg(v)\leq t$ and $\deg(u)\leq t-1$. The above claim then shows that
$(u,v)=f\cdot(\omega,\sigma)$ for some $f\in \efq[X]$. If
$\lm(u,v)\leq \lm(\omega,\sigma)$, this implies that $f$ is a
constant, as required. This also shows that
$\lm(u,v)=\lm(\omega,\sigma)$.  
\end{proof}

It will also be useful to recall that the uniqueness in the previous
proposition is an instance of a more general result.
\begin{proposition}\label{prop:unique}
For a field $K$ and for $\ell \in \bbNp$, let $\prec$ be any monomial
ordering on $K[X]^{\ell}$, and let $M\subseteq K[X]^{\ell}$ be any
$K[X]$-submodule. Suppose that both
$\bs{f}:=(f_1(X),\ldots,f_{\ell}(X))\in M\smallsetminus\{0\}$ and
$\bs{g}:=(g_1(X),\ldots,g_{\ell}(X))\in M\smallsetminus\{0\}$ have the
minimal leading monomial in $M\smallsetminus\{0\}$. Then there exists
a $c\in K^*$ such that $\bs{f}=c\cdot \bs{g}$.  
\end{proposition}

\begin{proof}
Suppose not. Since $\lm(\bs{f})=\lm(\bs{g})$, there exists a constant
$c\in K^*$ such that the leading monomial cancels in
$\bs{h}:=\bs{f}-c\bs{g}$. By assumption, $\bs{h}\neq \bs{0}$, and
$\lm(\bs{h})\prec\lm(\bs{f})$ -- a contradiction. 
\end{proof}

\subsection{\koetter{}'s \groebner{}-basis iteration}\label{sec:koetter}
Let us now recall the general form of \koetter{}'s 
iteration \cite{Koetter}, \cite{NH98}, as presented
by McEliece \cite[Sect.~VII.C]{McE03}.\footnote{
We have learned from Johan Rosenkilde that \cite{BL94}, \cite{BL97},
which predated \cite{Koetter}, already presented algorithms similar
to, and more general than \koetter{}'s iteration (see also
\cite[Sec.~2.6]{JNSV17}). For problems related to modules of vectors
of univariate polynomials, algorithms for 
computing the shifted (weak or canonical) {\it Popov form} of
$K[X]$-matrices have the lowest asymptotic complexity in some cases --
see, e.g., \cite{RS21} and the references therein for the case of
simultaneous Hermite--\pade{} approximation, and
\cite[Sec.~1.3.4]{N16} for the connection to \groebner{} bases. However,
for the fast Chase decoding algorithms considered in this paper, we
currently do not know if such methods will turn out to be more
efficient than \koetter{}'s iteration.}

Let $K$ be a field. For $\ell\in \bbNp$ and for a $K[X]$-submodule $M$
of $K[X]^{\ell+1}$ with $\rank(M)=\ell+1$, suppose that we have a
\groebner{} basis $G=\{\bs{g}_0,\ldots,\bs{g}_{\ell}\}$ for $M$ with
respect to some monomial ordering $\prec$ on $K[X]^{\ell+1}$. In such
a case, the leading monomials of the $\bs{g}_j$ must contain distinct
unit vectors,\footnote{For otherwise, the leading monomial of two
basis vectors would contain the same unit vector, so that the leading
monomial of one vector divides the leading monomial of the other
vector. In such a case, we may discard one of the basis vectors and
remain with a \groebner{} basis, which is, in particular, a set of
generators. So, we end up with a set of less than $\ell+1$ generators
for a free module of rank $\ell+1$ -- a contradiction (see, e.g., Ex.~11
on p.~32 of \cite{AtMa}).} and we may therefore assume w.l.o.g.~that
the leading monomial of $\bs{g}_j$ contains the $j$-th unit
vector, for all $j\in\{0,\ldots, \ell\}$ (where coordinates of vectors
are indexed by $0,\ldots, \ell$).  

Now let $D\colon K[X]^{\ell+1}\to K$ be a non-zero linear functional
that satisfies the following property:
\begin{description}

\item[{\bf MOD}] $\Mplus:= M\cap \ker(D)$ is a $K[X]$-module.

\end{description}

The purpose of \koetter{}'s iteration is to convert the $(\ell+1)$-element
\groebner{} basis\footnote{Where in this subsection, ``\groebner{} basis'' means a
\groebner{} basis with respect to $\prec$.} 
$G$ of $M$ to an $(\ell+1)$-element \groebner{} basis
$\Gplus=\{\bgplus_0,\ldots,\bgplus_{\ell}\}$ of $\Mplus$, while 
maintaining the property that $\lm(\bgplus_j)$ contains the 
$j$-th unit vector for all $j\in\{0,\ldots,\ell\}$. 

The following is a pseudo-code describing \koetter{}'s iteration.

\hfill\\ 

\begin{tabular}{|c|}
\hline
{\bf \koetter{}'s iteration without inversions}\\
\hline
\end{tabular}

\begin{description}

\item[\bf{Input}] A \groebner{} basis
$G=\{\bs{g}_0,\ldots,\bs{g}_{\ell}\}$ for the submodule $M\subseteq
\efq[X]^{\ell+1}$, with $\lm(\bs{g}_j)$ containing the $j$-th unit vector for
all $j$ 

\item[\bf{Output}] A \groebner{} basis
$\Gplus=\{\bgplus_0,\ldots,\bgplus_{\ell}\}$ for $\Mplus$ with
$\lm(\bgplus_j)$ containing the $j$-th unit vector for all $j$
(assuming {\bf{MOD}} holds)

\item[\bf{Algorithm}]

\begin{itemize} 

\item For $j=0, \ldots, \ell$, calculate $\Delta_j:=D(\bs{g}_j)$

\item Set $J:=\big\{j\in\{0,\ldots,\ell\}|\Delta_j\neq 0\big\}$

\item For $j \in\{0, \ldots, \ell\}\smallsetminus J$,

\begin{itemize}

\item  Set $\bgplus_j:=\bs{g}_j$

\end{itemize}

\item Let $j^*\in J$ be such that  $\lm(\bs{g}_{j^*})=\min_{j\in
J}\{\lm(\bs{g}_j)\}$ /* the leading monomials are distinct, and so 
$j^*$ is unique */ 

\item For $j\in J$ 

\begin{itemize}

\item If $j\neq j^*$

\begin{itemize}

\item Set $\bgplus_j := \Delta_{j^*}\bs{g}_j-\Delta_j \bs{g}_{j^*}$ 

\end{itemize}

\item Else /* $j=j^*$ */

\begin{itemize}

\item Set $\bgplus_{j^*} := \Delta_{j^*}\cdot X \bs{g}_{j^*}  -
D(X\bs{g}_{j^*})\cdot \bs{g}_{j^*}$\\ /* $=\big(\Delta_{j^*}\cdot X -
D(X\bs{g}_{j^*})\big)\bs{g}_{j^*}$ */

\end{itemize}

\end{itemize}

\end{itemize}

\end{description}

Note that for clarity of presentation, we have introduced a whole
new set of variables $\{\bs{g}_j^+\}$, although this is not really
necessary.  

\begin{proposition}\label{prop:koetter}
At the end of \koetter{}'s iteration, it holds that
$\Gplus=\{\bgplus_0,\ldots,\bgplus_{\ell}\}$ is a \groebner{} basis for 
$\Mplus$ and for all $j$, $\lm(\bgplus_j)$ contains the $j$-th unit 
vector. 
\end{proposition}

For a proof, see \cite[Sec.~VII.C]{McE03}.

\subsection{Fast Chase decoding on a tree}\label{sec:tree}
In the Chase-II decoding algorithm \cite[p.~173]{Chase} for decoding
a \emph{binary} code of minimum distance $d$, all possible $2^{\lfloor
d/2 \rfloor}$ error patterns on the $\lfloor d/2 \rfloor$ least
reliable coordinates are tested (i.e., subtracted from the received
word). For each tested error pattern, 
bounded distance decoding\footnote{Here, by {\it bounded distance
decoding} for a code of minimum distance $d$, we mean a decoding
algorithm that returns the unique codeword of distance up to $(d-1)/2$
from the received word (if exists), or declares failure otherwise.} is
performed, resulting in a list of up to  $2^{\lfloor d/2 \rfloor}$
candidate codewords. Finally, if the list is not empty, then the most
likely codeword is chosen from the list. 

For GRS codes over $\efq$, the type of Chase algorithm considered in
the current paper is the following variant of the Chase-II
algorithm. First, we assume a memoryless channel, e.g., as in
\cite[Sec.~III]{KV}. As in \cite{Chase}, we assume
that the decoder has probabilistic reliability information on the
received symbols. The $\eta$ least reliable coordinates are identified
for some pre-defined and (loosely speaking) small $\eta\in \bbNp$. Let
$\alpha_1,\ldots,\alpha_{\eta}$ be these least reliable coordinates
(where as usual, coordinates are labeled by elements of $\efq^*$), and
put $I:=\{\alpha_1,\ldots,\alpha_{\eta}\}$.  

Fix some $\mu\in\{1,\ldots,q\}$ and for each $i\in\{1,\ldots,\eta\}$,
let $A'_i\subset \efq$ be a subset of $\mu$ most probable choices for
the $\alpha_i$-th code symbol given the $\alpha_i$-th received
coordinate.\footnote{In the language of \cite{KV}, we look for $\mu$
largest values in the $\alpha_i$-th column of the {\it reliability
matrix}. For example, in \cite{BK10}, $\mu=2$.} Let
$a^{\star}$ be a symbol in $A'_i$ with the highest probability given
the $\alpha_i$-th received coordinate, and set
$A_i:=\{a-a^{\star}|a\in A'_i\}$. Hence $a^{\star}$ is the {\it
hard-decision (HD)} input to the decoder at coordinate $\alpha_i$ (an
entry of the vector $\bs{y}$ of Subsection \ref{sec:rs}),
while $A_i$ is a corresponding set of $\mu$ most probable errors
given the received symbol. 

Finally, fix some $\rmax\in \{1,\ldots,\eta\}$. The 
Chase decoding considered in the current paper runs over all test
error patterns on $I$ that are taken from  $A_1\times \cdots \times
A_{\eta}$ and have a Hamming weight of up to $\rmax$. For each such
error pattern, the algorithm performs (the equivalent of) bounded
distance decoding. Note that when $\rmax=\eta$, the test error patterns
are all the vectors in $A_1\times \cdots \times A_{\eta}$.

Let $B$ be the set of vectors of Hamming weight at most $\rmax$ in
$A_1\times \cdots \times A_{\eta}$. As in \cite{Wu12}, a directed tree
$T=T(\eta,I,\rmax,A_1,\ldots,A_{\eta})$ of depth $\rmax$ is constructed in the
following way. The root is the all-zero vector, and for all $r\in
\{1,\ldots,\rmax\}$, the  vertices at depth $r$ are the vectors in $B$ of
weight $r$.  

To define the edges of $T$, for each $r\geq
1$ and for each vertex $\bs{\beta}=(\beta_1,\ldots,\beta_{\eta})$ at
depth $r$ with non-zero entries at coordinates $i_1,\ldots,i_r$, we
pick a single vertex $\bs{\beta}'=(\beta'_1,\ldots,\beta'_{\eta})$ at
depth $r-1$ that is equal to $\bs{\beta}$ on all coordinates, except
for one $i_{\ell}$ ($\ell\in\{1,\ldots,r\}$), for which
$\beta'_{i_{\ell}}=0$. Note that given $\bs{\beta}$, there are $r$
distinct ways to choose $\bs{\beta}'$, and we simply fix one such
choice of $\bs{\beta}'$ for each $\bs{\beta}$. Now the edges of $T$
are exactly all such pairs $(\bs{\beta}',\bs{\beta})$ (see Figure
\ref{fig:tree} for an example).

Note that the edge $(\bs{\beta}',\bs{\beta})$ defined above corresponds
to adding exactly one additional modified coordinate, namely, coordinate
$\alpha_{i_{\ell}}$, in which the assumed error value is
$\beta_{i_{\ell}}$. Hence, the edge $(\bs{\beta}',\bs{\beta})$ can be
identified with the pair
$(\bs{\beta'},(\alpha_{i_{\ell}},\beta_{i_{\ell}}))$. Similarly, a 
path from the root to a vertex at depth $r\geq 1$ (and hence the 
vertex itself) can be identified with a sequence 
$$
\big((\alpha_{i_{1}},\beta_{i_{1}}),\ldots,
(\alpha_{i_{r}},\beta_{i_{r}})\big)\in((\efq^*)^2)^r
$$
for which the $\alpha_{i_{\ell}}$'s are distinct.

\begin{figure}[h!]
$$
{\scriptsize
\xymatrix@C=1.7em{
  &   &    &   &   & 00000\ar[dllll] \ar[dl] \ar[drr] \ar[drrr]
\ar[drrrr]\\ 
  & 10000 \ar[dl] \ar[d] \ar[dr] &    &   & 01000 \ar[dl] \ar[d]
\ar[dr]  \ar[drr] &  &  & 00100  \ar[d]     & 00010 \ar[d] & 00001
\ar[d]\\  
10100 \ar[d]& 10010\ar[d] & 10001\ar[d] & 01001\ar[d] & 01010\ar[d] &
11000\ar[d] & 01100\ar[d] & 00110\ar[d] & 00011\ar[d] 
& 00101 \ar[d]\\
10101 & 11010 & 11001 & 01011 \ar[d]& 01110\ar[d] & 11100\ar[d] &
01101\ar[d] & 10110 & 10011 \ar[d] & 00111 \\
 & & & 11011 & 11110 \ar[d] & 11101 & 01111 & & 10111\\
 & & & & 11111
}
}
$$
\caption{Example of the tree $T$ for $\rmax=\eta=5$,
$A_1=\cdots=A_5=\{0,1\}$. The vertices at each depth
$r\in\{0,\ldots,\rmax=5\}$ are the vectors of weight $r$, and the
edges are chosen such that for each vertex at depth $r\geq 1$, we
connect exactly one vertex from depth $r-1$ that is obtained by
transforming one non-zero value to zero. There are $r$ different ways
to do this, but we simply fix one of them. For example, in the figure,
$11011$ at depth $4$ is connected to $01011$ at depth
$3$. Alternatively, it could be connected to either one of $10011$,
$11001$, $11010$.}
\label{fig:tree}
\end{figure}
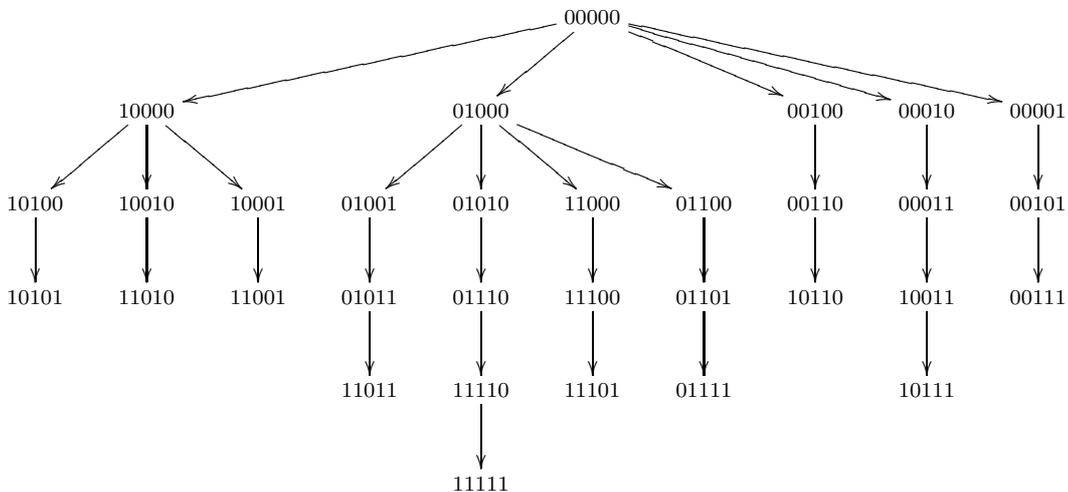

The main ingredient of Wu's fast Chase algorithm, as well as
of the algorithm of the current paper, is an efficient algorithm for
updating the ELP (and additional polynomials) for adding a single
modified coordinate $\alpha_{i_r}$ and the corresponding error value,
$\beta_{i_r}$. The tree $T$ is then traversed depth
first, saving intermediate results on vertices whose out degree is
larger than $1$, and applying the polynomial-update algorithm on the
edges. Because the tree is traversed depth first and has depth $\rmax$,
there is a need to save at most $\rmax$ vertex calculations at each
time (one for each depth).\footnote{We thank I.~Tamo for pointing this
out.} See ahead for details.     

\section{The minimization problem for fast Chase decoding}
\label{sec:module} 
Wu's LFSR minimization problem $\bs{A}[\bs{\sigma}_i]$
\cite[p.~112]{Wu12} is defined over an $\efq$-vector space
of pairs of polynomials that in general is not an
$\efq[X]$-module. The key observation is that using Forney's formula, 
Wu's minimization problem can be replaced by a  minimization 
problem over an $\efq[X]$-module. 

\begin{remarknn}
{\rm
For simplicity, we will assume from this point on that $d$ is
odd, so that $d=2t+1$. It is straightforward to modify the following
derivation for the case of even $d$.
}
\end{remarknn}

\begin{definition}\label{def:mr}
{\rm
For $r\in \{0,\ldots,n\}$, distinct $\alpha_1,\ldots,\alpha_r\in
\efq^*$,  and $\beta_1,\ldots,\beta_r\in \efq^*$ (not necessarily
distinct), let   
$$
M_r=M_r(\synd{y},\alpha_1,\ldots,\alpha_r,\beta_1,\ldots,\beta_r)
$$
be the set of all pairs $(u,v)\in \efq[X]^2$ satisfying the following
conditions:
\begin{enumerate}

\item $u\equiv \synd{y}v \mod (X^{d-1})$

\item $\forall j\in\{1,\ldots,r\}$, $v(\alpha_j^{-1})=0$ and 
$\beta_j a_j v'(\alpha_j^{-1})= -\alpha_j u(\alpha_j^{-1})$ (with
$a_j:=\tl{a}_{j'}$ for the $j'$ with $\alpha_j=\lambda^{j'}$).

\end{enumerate}
}
\end{definition}

The possibility of using \koetter{}'s iteration as an
alternative to Wu's method follows almost immediately 
from the following theorem. 

\begin{theorem}\label{thm:main}

\begin{enumerate}

\item For all $r$, $M_r$ is an $\efq[X]$-module. 

\item With the terminology of the previous section, if
$\dham(\bs{y},\bs{x})\leq t+r$, $\alpha_1,\ldots,\alpha_r$ are 
error locations and $\beta_1,\ldots,\beta_r$ are the corresponding
error values, then $(\omega,\sigma)\in M_r$ and
$$
\lm(\omega,\sigma)=\min\big\{\lm(u,v)|(u,v)\in
M_r\smallsetminus\{0\}\big\}. 
$$

\end{enumerate}

\end{theorem}

\begin{proof}

1. $M_r$ is clearly an
$\efq$-vector space. For $f(X)\in 
\efq[X]$ and $(u,v)\in M_r$, we would like to show that
$f\cdot(u,v)\in M_r$. Clearly, $(fu,fv)$ satisfies the required
congruence, and also $fv$ has the required roots. It remains to verify
that for all $j$, $\beta_j a_j(fv)'(\alpha_j^{-1}) = -\alpha_j
(fu)(\alpha_j^{-1})$. Now,   
\begin{multline*}
(fv)'(\alpha_j^{-1})=(f'v)(\alpha_j^{-1})+(fv')(\alpha_j^{-1}) =
(fv')(\alpha_j^{-1}) \\ = f(\alpha_j^{-1})\cdot
\frac{-\alpha_j}{\beta_j a_j}u(\alpha_j^{-1}) =
-\frac{\alpha_j}{\beta_j a_j}(fu)(\alpha_j^{-1}),  
\end{multline*}
where in the second equality we used $v(\alpha_j^{-1})=0$ and in the
third equality we used $\beta_j a_j v'(\alpha_j^{-1})= -\alpha_j
u(\alpha_j^{-1})$ (note that $\beta_j a_j\neq 0$). 

2. First, $(\omega,\sigma)\in M_r$ by the key equation (\ref{eq:key})
and Forney's formula (\ref{eq:forney}). The proof of minimality is by
induction on $r$. For $r=0$, the assertion is just 
Proposition \ref{prop:rzero}. Suppose that $r\geq 1$, and the
assertion holds for $r-1$. Let $\tlbsy$ be obtained from $\bs{y}$ by
subtracting $\beta_r$ from coordinate $\alpha_r$. Let
$\tlsig:=\sigma/(1-\alpha_rX)$ (the error locator 
for $\tlbsy$) and let $\tlom$ be the error evaluator for $\tlbsy$. By
the induction hypothesis,   
\begin{equation}\label{eq:min}
\lm(\tlom,\tlsig)=\min\big\{\lm(u,v)|(u,v)\in \tilde{M}_{r-1}\big\},
\end{equation}
with 
$$
\tilde{M}_{r-1} :=
M_{r-1}(\syndtly,\alpha_1,\ldots,\alpha_{r-1},\beta_1,\ldots,\beta_{r-1}).
$$

To continue, we will need a lemma.

\begin{lemmann}
For $(u,v)\in M_r$, write $\tl{v}:=v/(1-\alpha_r X)$ and put
$h:=u - \beta_r a_r\tl{v}$. Then 
$(1-\alpha_rX)|h(X)$. Moreover, writing $\tl{h}:=h/(1-\alpha_rX)$, the
map $\psi\colon (u,v)\mapsto (\tl{h},\tl{v})$ maps $M_r$ into
$\tilde{M}_{r-1}$, and satisfies
$\psi(\omega,\sigma)=(\tlom,\tlsig)$.  
\end{lemmann}

\begin{proof}[Proof of Lemma]
Since $v=(1-\alpha_rX)\tl{v}$, we have
$v'=-\alpha_r\tl{v}+(1-\alpha_rX)\tl{v}'$, and therefore
$v'(\alpha_r^{-1})=-\alpha_r\tl{v}(\alpha_r^{-1})$. Hence,
\begin{eqnarray*}
h(\alpha_r^{-1}) & = & u(\alpha_r^{-1})-\beta_r a_r\tl{v}(\alpha_r^{-1})
\\
 & = & -\frac{\beta_r a_r}{\alpha_r} v'(\alpha_r^{-1})-\beta_r a_r
\tl{v}(\alpha_r^{-1})=0, 
\end{eqnarray*}
which proves the first assertion. 

For the second assertion, note first
that by (\ref{eq:oneloc}),
$$
\syndtly\equiv\synd{y} - \frac{\beta_r a_r}{1-\alpha_r X}\mod (X^{d-1}),
$$
and therefore
\begin{eqnarray*}
\syndtly \tl{v} & \equiv & \synd{y}\tl{v}-\frac{\beta_r
a_r}{1-\alpha_r X}\tl{v}\\ 
 & = & \frac{1}{1-\alpha_r X}(\synd{y}v-\beta_r a_r\tl{v})\\
 & \equiv & \frac{1}{1-\alpha_r X}(u-\beta_r a_r\tl{v}) = \tl{h}, 
\end{eqnarray*}
(where ``$\equiv$'' stands for congruence modulo $X^{d-1}$), which
implies that $(\tl{h},\tl{v})$ satisfies the required congruence
relation in the definition of $\tilde{M}_{r-1}$. Also, clearly
$\tl{v}(\alpha_j^{-1})=0$ for all $j\in \{1,\ldots,r-1\}$. Finally,
using $v'=-\alpha_r\tl{v}+(1-\alpha_rX)\tl{v}'$ again, we see that for
all $j\in\{1,\ldots,r-1\}$, 
$$
\tl{v}'(\alpha_j^{-1}) =
\frac{v'(\alpha_j^{-1})}{1-\alpha_r\alpha_j^{-1}}
 = -\frac{\alpha_j}{\beta_j a_j}\cdot
\frac{u(\alpha_j^{-1})}{1-\alpha_r\alpha_j^{-1}}  
 =  -\frac{\alpha_j}{\beta_j a_j}\cdot \tl{h}(\alpha_j^{-1}).
$$
This proves that $\psi$  maps $M_r$ into $\tilde{M}_{r-1}$.

Finally, we have $\psi(\omega,\sigma)=(\tl{h},\tlsig)$ with
$\tl{h}=(\omega-\beta_r a_r \tlsig)/(1-\alpha_r X)$, and it remains to
verify that $\tl{h}=\tlom$.\footnote{Note that if the total number of
errors is at most $d-1$, then it is clear from the above that
$\tl{h}=\tlom$, as both are congruent to $\syndtly\tlsig$
modulo $X^{d-1}$ and have a degree $\leq d-2$. However, the following
proof does not require this assumption.} 
Let $\alpha'_1,\ldots,\alpha'_{\veps}\in \efq^*$ be
some enumeration of all error locators, let
$\beta'_1,\ldots,\beta'_{\veps}\in \efq^*$
be the corresponding error values, and let $a'_1,\ldots,a'_{\veps}$ be
the corresponding entries of the vector $\tl{\bs{a}}$. Assume
w.l.o.g.~that $\alpha'_{\veps}=\alpha_r$ (and hence
$\beta'_{\veps}=\beta_r$ and $a'_{\veps}=a_r$). Then
\begin{eqnarray*}
\tl{h} & = & \frac{\omega-\beta_r a_r \tlsig}{1-\alpha_r X} \\
 & = &
\frac{1}{1-\alpha'_{\veps}X}\Big(\sum_{i=1}^{\veps}\beta'_ia'_i\prod_{j=1,j\neq
i}^{\veps}(1-\alpha'_j X)
-\beta'_{\veps}a'_{\veps}\prod_{j=1}^{\veps-1}(1-\alpha'_j X)\Big)\\ 
 & = &
\frac{1}{1-\alpha'_{\veps}X}\Big(\sum_{i=1}^{\veps-1}\beta'_ia'_i\prod_{j=1,j\neq
i}^{\veps}(1-\alpha'_j X)\Big) \\
 & = & \sum_{i=1}^{\veps-1}\beta'_ia'_i\prod_{j=1,j\neq
i}^{\veps-1}(1-\alpha'_j X)=\tlom.
\end{eqnarray*}

\end{proof}

Returning to the proof of part 2 of the theorem, if $(u,v)\in M_r$ and
$v=c\cdot\sigma$ for some $c\in \efq^*$, 
then we must have $\lm(u,v)\geq (0,X^{\deg(\sigma)}) =
\lm(\omega,\sigma)$. Let us therefore take 
$(u,v)\in M_r\smallsetminus\{0\}$ with $v\neq c\sigma$ for all
$c\in\efq^*$. Then also 
$\psi(u,v)\neq c(\tlom,\tlsig)$ for all $c\in\efq^*$, and hence, 
by the induction hypothesis, the lemma, and Proposition
\ref{prop:unique},  
\begin{equation}\label{eq:ge}
\lm(\psi(u,v))>\lm(\tlom,\tlsig)=(0,X^{\deg(\sigma)-1}).
\end{equation}
If the leading monomial of $\psi(u,v)$ is of the form $(0,X^{\ell})$
for some $\ell$, then $\lm(\psi(u,v))=(0,X^{\deg(v)-1})$, and
(\ref{eq:ge}) implies $\deg(v)>\deg(\sigma)$, so that certainly
$\lm(u,v)>\lm(\omega,\sigma)$.   

Suppose therefore that $\lm(\psi(u,v))$ is of the form $(X^{\ell},0)$
for some $\ell$, that is, 
$\lm(\psi(u,v))=(X^{\deg(h)-1},0)$. In this case, (\ref{eq:ge})
implies that $\deg(h)-1>\deg(\sigma)-2$, that is, $\deg(h)\geq
\deg(\sigma)$. But since $h=u-\beta_r a_r \tl{v}$, this implies that at
least one of $u$ and $\tl{v}$ must have a degree that is at least as
large as $\deg(\sigma)$. Now, if $\deg(u)\geq \deg(\sigma)$, that
is,
if $\deg(u)>\deg(\sigma)-1$, then
$\lm(u,v)>\lm(\omega,\sigma)=(0,X^{\deg(\sigma)})$. Similarly, if
$\deg(\tl{v})\geq\deg(\sigma)$, then $\deg(v)>\deg(\sigma)$, and again
$\lm(u,v)>\lm(\omega,\sigma)$. This completes the proof of Theorem
\ref{thm:main}. 
\end{proof}

When moving from
$M_r:=M_r(\synd{y},\alpha_1,\ldots,\alpha_r,\beta_1,\ldots,\beta_r)$ to
$M_{r+1}:= M_{r+1}(\synd{y}, \alpha_1, \ldots, \alpha_{r+1},
\beta_1,\ldots, \beta_{r+1})$, two additional functionals must 
be zeroed. It was already proved in the theorem that each $M_r$ is an 
$\efq[X]$-module. Also, the intersection of $M_r$ with the set
of pairs $(u,v)$ for which $v(\alpha_{r+1}^{-1})=0$ is clearly an
$\efq[X]$-module. Hence, if each ``root condition'' comes before the
corresponding ``derivative condition,''  we may use \koetter{}'s
iteration twice in order to move from a \groebner{} basis for $M_r$ to a 
\groebner{} basis for $M_{r+1}$. 

A detailed description of the application of \koetter{}'s iteration for
moving from $M_r$ to $M_{r+1}$ appears in the following
subsection. This is the algorithm carried out on the edges of the tree
$T$ of Section \ref{sec:tree}. 

For initiating the fast Chase algorithm on the root of $T$, we need a
\groebner{} basis for $M_0$.\footnote{Note that by
Proposition \ref{prop:rzero}, any algorithm that finds a \groebner{}
basis for $M_0$ can also be used for bounded-distance decoding.} Several
algorithms for achieving this goal appear in \cite{Fitz95}. In
particular, Algorithm 4.3 of \cite{Fitz95} is the Euclidean
algorithm,\footnote{The stopping condition of \cite[Alg.~4.3]{Fitz95}
assures that throughout its run, the leading monomials of both
processed pairs of polynomials contain $(1,0)$. Hence, the division
algorithm used in the algorithm effectively divides the two scalar
polynomials on the first coordinate, and performs the same
calculations as the Euclidean algorithm.}
while Algorithm 4.7 of \cite{Fitz95} is similar in nature to the
BM algorithm.  

In fact, we remark that the BM algorithm \emph{itself} can be used to
obtain a \groebner{} basis for $M_0$. Informally, after running the BM
algorithm, two pairs of polynomials are obtained from the two
polynomials updated during the algorithm, and then at most one
additional leading monomial cancellation is required for obtaining the
desired \groebner{} basis. Since the proof is rather technical and this
is outside the main scope of the current paper, we will not elaborate
on this issue.

\section{Algorithms}\label{sec:algs}

\subsection{The basic algorithm on an edge of the decoding tree}\label{sec:basic}
Using the terminology of Section \ref{sec:koetter}, in the current
context we have $\ell=1$, and, as 
already mentioned, we have two types of \koetter{} iterations: one for a
root condition, and the other for a derivative condition. For
convenience, we will use here a version of \koetter{}'s iteration that
includes inversions. In this version, the right-hand sides of the
update rules are both divided by $\Delta_{j^*}$ (multiplication of
elements by non-zero constants obviously takes a \groebner{} basis to a
\groebner{} basis). 

In the $r$-th {\it root iteration}, the linear functional $D$ of
\koetter{}'s iteration acts on a pair
$(u,v)$ as $D(u,v)=v(\alpha_r^{-1})$, and hence on $X\cdot(u,v)$ as
$D(X\cdot(u,v))=\alpha_r^{-1}D(u,v)$. In the $r$-th {\it derivative
iteration} (which must come after the $r$-th root
iteration), we have 
$$
D(u,v)=\beta_r a_rv'(\alpha_r^{-1})+\alpha_r u(\alpha_r^{-1}),
$$
and therefore also
\begin{eqnarray*}
D(X\cdot(u,v))  & = & \beta_r a_r(Xv)'(\alpha_r^{-1})+\alpha_r
(Xu)(\alpha_r^{-1}) \\
 & = & \beta_r a_r \alpha_r^{-1} v'(\alpha_r^{-1}) +
u(\alpha_r^{-1})\\
 & = & \alpha_r^{-1} D(u,v),
\end{eqnarray*}
where in the second equality we used $(Xv)'=Xv'+v$ and
$v(\alpha_r^{-1})=0$. So, for both types of iterations, we have
$D(X\cdot(u,v))/D(u,v)=\alpha_r^{-1}$ if $D(u,v)\neq 0$. Hence, the 
iteration corresponding to a single location $\alpha_r$ has the
following form. 

Note that the above root and derivative iterations correspond to the
values $\broot,\bder$ (resp.) of the variable $\tau$ in Algorithm A.

\hfill\\ 

\begin{tabular}{|c|}
\hline
{\bf Algorithm A: \koetter{}'s iteration for adjoining error location $\alpha_r$}\\
\hline
\end{tabular}

\begin{description}

\item[\bf{Input}] 

\begin{itemize}

\item A \groebner{} basis
$G=\{\bs{g}_0=(g_{00},g_{01}),\bs{g}_1=(g_{10},g_{11})\}$ 
for
$M_{r-1}(\synd{y},\alpha_1,\ldots,\alpha_{r-1},\beta_1,\ldots,\beta_{r-1})$,
with $\lm(\bs{g}_j)$ containing the $j$-th unit vector for 
$j\in\{0,1\}$ 

\item The next error location, $\alpha_r$, and the corresponding error
value, $\beta_r$

\end{itemize}

\item[\bf{Output}] A \groebner{} basis
$\Gplus=\{\bgplus_0=(\gplus_{00},\gplus_{01}),\bgplus_{1}=(\gplus_{10},\gplus_{11})\}$
for \\ $M_r(\synd{y},\alpha_1,\ldots,\alpha_r,\beta_1,\ldots,\beta_r)$
with $\lm(\bgplus_j)$ containing the $j$-th unit vector for $j\in\{0,1\}$ 

\item[\bf{Algorithm}]

\begin{itemize} 

\item For {\bf type} $=$ {\bf root}, {\bf der}

\begin{itemize}

\item If {\bf type} $=$ {\bf der}, 

\begin{itemize}

\item For $j = 0, 1$, set $\bs{g}_j:=\bgplus_j$ /* init:
output of {\bf root} iter.~*/

\end{itemize}

\item For $j=0, 1$, calculate 
$$
\Delta_j:=\begin{cases}
g_{j1}(\alpha_r^{-1}) & \text{if {\bf type}$=${\bf root}}\\
\beta_ra_r g_{j1}'(\alpha_r^{-1})+\alpha_r g_{j0}(\alpha_r^{-1}) &
\text{if {\bf type}$=${\bf der}}
\end{cases}
$$

\item Set $J:=\big\{j\in\{0,1\}|\Delta_j\neq 0\big\}$

\item For $j\in\{0,1\}\smallsetminus J$, set $\bgplus_j:=\bs{g}_j$

\item Let $j^*\in J$ be such that  $\lm(\bs{g}_{j^*})=\min_{j\in 
J}\{\lm(\bs{g}_j)\}$  

\item For $j\in J$ 

\begin{itemize}

\item If $j\neq j^*$

\begin{itemize}

\item Set $\bgplus_j := \bs{g}_j-\frac{\Delta_j}{\Delta_{j^*}}
\bs{g}_{j^*}$ 

\end{itemize}

\item Else /* $j=j^*$ */

\begin{itemize}

\item Set $\bgplus_{j^*} := (X-\alpha_r^{-1})\bs{g}_{j^*}$

\end{itemize}

\end{itemize}

\end{itemize}

\end{itemize}

\end{description}

Again, for clarity of presentation, we have introduced a whole
new set of variables $\{\bgplus_j\}$, although this is not really
necessary. 

If a successive application of the algorithm down the
path from the root to a vertex
$\big((\alpha_{i_{1}},\beta_{i_{1}}),\ldots,
(\alpha_{i_{r-1}},\beta_{i_{r-1}})\big)$ 
of $T$ results in \groebner{} basis for 
$M_{r-1}(\synd{y},\alpha_{i_1},\ldots,
\alpha_{i_{r-1}},\beta_{i_1},\ldots,\beta_{i_{r-1}})$, 
then an additional application on the edge
$(\alpha_{i_{r}},\beta_{i_{r}})$ will result in a \groebner{} basis for 
$M_{r}(\synd{y},\alpha_{i_1},\ldots, 
\alpha_{i_{r}},\beta_{i_1},\ldots,\beta_{i_{r}})$. 

It therefore follows from Theorem \ref{thm:main} that if a vertex  
$$
\bs{v}:=\big((\alpha_{i_{1}},\beta_{i_{1}}),\ldots,
(\alpha_{i_{r}},\beta_{i_{r}})\big) 
$$
of $T$ is a ``direct hit,'' in the sense that
$\alpha_{i_{1}},\ldots,\alpha_{i_{r}}$ are indeed error locations
with respective error values $\beta_{i_{1}},\ldots,\beta_{i_{r}}$,
and if $\veps\leq t+r$, then the second element of the \groebner{} basis
on $\bs{v}$ is $c\cdot (\omega,\sigma)$ for some non-zero
$c$.\footnote{Recall that $\lm(\omega,\sigma)$ contains the
unit vector $(0,1)$.} 

While not necessary for the correctness of the algorithm, it is of
interest to consider the case where, although the tested error pattern
is not a direct hit, the difference between the number of correct indices
and incorrect indices is at least $\veps-t$. For this case, see
Appendix \ref{app:indirect} 

Two faster versions of Algorithm A appear in Appendix
\ref{app:simplifications}: in the first, two \emph{polynomials} (rather
than two \emph{pairs} of polynomials) are maintained, and in the second,
which is even more efficient, the algorithm works with low-degree
polynomials.  

\begin{remark}\label{remark:double}
{\rm
At a first glance, it may seem that the need to use two stages (root
and derivative iterations)
comes at the cost of doubling the complexity in comparison to
\cite[Alg.~1]{Wu12}. However, this is not the case: As shown ahead in
Appendix \ref{app:simplifications}, for the variant of Algorithm A
described in Section \ref{app:pairs}, the complexity of Wu's
algorithm is lower only by a factor about $5/6$ (or $10/11$ in
characteristic $2$) than Algorithm A. See also Remark
\ref{rem:doubleev} ahead for Algorithm B of the following section. We
also note that  Algorithm C of Section \ref{app:lowdeg}, which is
another variant of Algorithm A, has a lower complexity than
\cite[Alg.~1]{Wu12}.  
}
\end{remark}

\subsection{Working with vectors of evaluations}\label{sec:eval}
As already mentioned, to achieve a complexity of $O(n)$ per 
modified symbol, one can use \koetter{}'s method of updating vectors of
evaluations.  Whereas in \cite{Wu12} this requires a
complicated modification of the original algorithm in order to avoid
syndrome updates, it is straightforward to modify Algorithm A to an
``evaluated'' version. 

In Algorithm B below, for a fixed primitive element $\lambda\in
\efq^*$ we let
$\bs{\lambda}:=(1,\lambda^{-1},\lambda^{-2},\ldots,\lambda^{-(q-2)})$.
Also, for a polynomial $f\in \efq[X]$, we let
$f(\bs{\lambda}):=(f(1),f(\lambda^{-1}),f(\lambda^{-2}),\ldots,f(\lambda^{-(q-2)}))$.
Finally, in the algorithm below, $-\odot-$ stands for component-wise 
multiplication of vectors, that is, 
$$
(v_1,v_2,\ldots,v_{\ell})\odot(u_1,u_2,\ldots,u_{\ell}):=(v_1u_1,v_2u_2,\ldots,
v_{\ell}u_{\ell})
$$
(where $\ell\in\bbNp$ and the $u_i,v_i$ are taken from some ring).

Note that the algorithm requires tracing the evaluation vectors of 
the four polynomials implicit in the \groebner{} basis, as
well as the evaluation vectors of the formal derivatives of two of
these four polynomials. 

\hfill\\ 

\begin{tabular}{|c|}
\hline
{\bf Algorithm B: Adjoining error location
$\alpha_r$}\\{\bf for vectors of evaluations, complexity $O(n)$}\\
\hline
\end{tabular}

\begin{description}

\item[\bf{Input}] 

\begin{itemize}

\item For a \groebner{} basis
$G=\{\bs{g}_0=(g_{00},g_{01}),\bs{g}_1=(g_{10},g_{11})\}$ for
$M_{r-1}(\synd{y},\alpha_1,\ldots,\alpha_{r-1},
\beta_1,\ldots,\beta_{r-1})$ 
with $\lm(\bs{g}_j)$ containing the $j$-th unit vector for 
$j\in\{0,1\}$, the input includes the following data:
$$
\gamma_j:=(\bs{v}_{j0},\bs{v}_{j1},\bs{v}_{j2},\bs{m}_j) :=
\big(g_{j0}(\bs{\lambda}),g_{j1}(\bs{\lambda}), 
g_{j1}'(\bs{\lambda}),\lm(\bs{g}_j)\big), j=0,1 
$$

\item The next error location, $\alpha_r$, and the
corresponding error value, $\beta_r$

\end{itemize}

\item[\bf{Output}] For some \groebner{} basis
$G^+=\{\bgplus_0=(g^+_{00},g^+_{01}),\bgplus_1=(g^+_{10},g^+_{11})\}$ 
for $M_{r}(\synd{y},\alpha_1,\ldots,\alpha_{r},\beta_1,\ldots,\beta_{r})$
with $\lm(\bgplus_j)$ containing the $j$-th unit vector for 
$j\in\{0,1\}$, the output consists of the following data:
$$
\gamma^+_j:=(\bs{v}^+_{j0},\bs{v}^+_{j1},\bs{v}^+_{j2},\bs{m}^+_j) := 
\big(g^+_{j0}(\bs{\lambda}),g^+_{j1}(\bs{\lambda}), 
(g^+_{j1})'(\bs{\lambda}),\lm(\bgplus_j)\big),j=0,1
$$

\item[\bf{Algorithm}]

\begin{itemize} 

\item For {\bf type} $=$ {\bf root}, {\bf der}

\begin{itemize}

\item If {\bf type} $=$ {\bf der}, 

\begin{itemize}

\item For $j = 0, 1$, set $\gamma_j:=\gamma^+_j$ /* init:
output of {\bf root} iter.~*/

\end{itemize}

\item For $j=0, 1$, calculate (using appropriate entries of
$\bs{v}_{j0},\bs{v}_{j1},\bs{v}_{j2})$ 
$$
\Delta_j:=\begin{cases}
g_{j1}(\alpha_r^{-1}) & \text{if {\bf type}$=${\bf root}}\\
\beta_ra_r g_{j1}'(\alpha_r^{-1})+\alpha_r g_{j0}(\alpha_r^{-1}) &
\text{if {\bf type}$=${\bf der}}
\end{cases}
$$

\item Set $J:=\big\{j\in\{0,1\}|\Delta_j\neq 0\big\}$

\item For $j\in\{0,1\}\smallsetminus J$, set $\gamma^+_j:=\gamma_j$

\item Let $j^*\in J$ be such that  $\bs{m}_{j^*}=\min_{j\in
J}\{\bs{m}_j\}$  

\item For $j\in J$ 

\begin{itemize}

\item If $j\neq j^*$

\begin{itemize}

\item For $i=0,1,2$, set
$\bs{v}^+_{ji}:=\bs{v}_{ji}-\frac{\Delta_j}{\Delta_{j^*}}\bs{v}_{j^*i}$

\item Set $\bs{m}^+_j:=\bs{m}_j$ and put
$\gamma^+_j:=(\bs{v}^+_{j0},\bs{v}^+_{j1},\bs{v}^+_{j2},\bs{m}^+_j)$ 

\end{itemize}

\item Else /* $j=j^*$ */

\begin{itemize}

\item For $i=0,1$, set
$$
\bs{v}^+_{ji} :=
\big(\bs{\lambda}-(\alpha_r^{-1},\alpha_r^{-1},\ldots,\alpha_r^{-1})\big)\odot
\bs{v}_{ji}
$$

\item Set /* using
$\big[(X-\alpha_r^{-1})g_{j1}\big]'=(X-\alpha_r^{-1})g'_{j1}+g_{j1}$ */
$$
\bs{v}^+_{j2}:=\big(\bs{\lambda}-(\alpha_r^{-1},\alpha_r^{-1},\ldots,\alpha_r^{-1})\big)
\odot \bs{v}_{j2} + \bs{v}_{j1}
$$
 
\item Set $\bs{m}^+_j:=X\cdot \bs{m}_j$ and put
$\gamma^+_j:=(\bs{v}^+_{j0},\bs{v}^+_{j1},\bs{v}^+_{j2},\bs{m}^+_j)$ 

\end{itemize}

\end{itemize}

\end{itemize}

\end{itemize}

\end{description}

\begin{remark}\label{rem:doubleev}
{\rm
\begin{enumerate}

\item Algorithm B maintains a total of $6$ evaluation vectors, and its
complexity is dominated by a total of $2\cdot 6=12$ per-coordinate
multiplications of evaluation vectors. Hence, the total complexity on
one edge is $12n$ finite-field multiplications. 

\item As opposed to \cite[Alg.~1]{Wu12}, in \cite[Alg.~2]{Wu12} there is an
explicit equivalent to the two stages (root and derivative) of
Algorithm B: in each application of \cite[Alg.~2]{Wu12}, there is one
stage called ``2) Updating'', followed by a stage called ``3)
Converting''. Just like Algorithm B, \cite[Alg.~2]{Wu12} maintains $6$
evaluation vectors, where in each of the Updating and Converting
stages only $4$ of them are updated. However, the complexity depends
on the total number of per-coordinate multiplications of an evaluation
vector, and not on the number of updated vectors. This number of 
per-coordinate multiplications depends on the case in
the Updating stage. For example, for Case 3, it seems that there are
$8$ distinct per-coordinate multiplications in the Updating stage,
followed by $4$ per-coordinate multiplications in the Converting
stage. This gives a total of $12$ such multiplications, just as in
Algorithm B. On the other hand, the number of per-coordinates
multiplications appears to be higher for Case 8 of the Updating
stage. All-in-all, it seems to be fair to say that the two algorithms
have a similar complexity. 

\end{enumerate}

}
\end{remark}

\subsection{High-level description of the decoding algorithm}\label{sec:high}
Let us now describe the high-level flow of the decoding algorithm.

\begin{enumerate} 

\item Perform bounded-distance HD decoding. If this decoding finds a
codeword within Hamming distance $t$ from the received word, output
this codeword and  exit. Otherwise, proceed to the fast Chase decoding
algorithm. 

\item Find a \groebner{} basis
$\{\bs{g}_0=(g_{00},g_{01}),\bs{g}_1=(g_{10},g_{11})\}$ for $M_0$ with
$\lm(\bs{g}_0)$ containing $(1,0)$ and $\lm(\bs{g}_1)$ containing
$(0,1)$. As shown in \cite{Fitz95}, this can be done with
an equivalent of any of the standard bounded-distance HD decoding
algorithms, and can also be used for HD decoding in Step
1.

\item Calculate the derivatives $g'_{01},g'_{11}$, and evaluate
polynomials to obtain 
$$
\gamma_j:=\big(g_{j0}(\bs{\lambda}),g_{j1}(\bs{\lambda}), 
g_{j1}'(\bs{\lambda}),\lm(\bs{g}_j)\big), j=0,1.
$$
Store $\gamma_0,\gamma_1$ in the memory for depth $0$.

\item Using reliability information, identify a set
$I=\{\alpha_1,\ldots,\alpha_{\eta}\}$ of $\eta$ least reliable
coordinates . For each $\alpha\in I$, find $A_{\alpha}$, the set
of $\mu$ most probable HD errors for the $\alpha$-th coordinate given the
$\alpha$-th received symbol.  Together with the pre-defined depth,
$\rmax$, this completely determines the tree $T$ of Section
\ref{sec:tree}. 

\item Traverse the tree $T$ depth first.  When visiting an edge
$(\bs{u'},\bs{u})$ between  a vertex  $\bs{u'}$ at depth $r-1$ and a
vertex $\bs{u}$ at depth $r$:  

\begin{itemize}

\item Perform Algorithm B, taking the inputs $\gamma_0,\gamma_1$ from
the memory for depth $r-1$. 

\item Store the outputs $\gamma^+_0,\gamma^+_1$
in the memory for depth $r$. 

\item If the following conditions
hold:

\begin{itemize}

\item $\bs{v}^+_{1,1}$ has exactly $t+r$ zero entries, and 

\item $\bs{m}_1^+=(0,X^{t+r})$ (this is equivalent to
$\deg(g^+_{1,1})=t+r$, as $\lm(\bgplus_1)$ contains $(0,1)$) 

\end{itemize}

then:

\begin{itemize}

\item  Letting $i_1,\ldots,i_{t+r}$ be the indices of zero entries of
$\bs{v}^+_{1,1}$ (counting indices from $0$), let the error locators
be $\alpha_{j}:=\lambda^{i_j}$, $j=1,\ldots,t+r$. Calculate 
the corresponding error values using appropriate entries of
$\bs{v}^+_{1,0}$ (evaluation vector of $g^+_{10}$) and $\bs{v}^+_{1,2}$
(evaluation vector of $(g^+_{11})'$) by the Forney formula
(\ref{eq:forney}):\footnote{Note that in the expression for
$\beta_{j}$, the denominator is non-zero because $g^+_{11}$ is
separable by the above assumptions.}  
$$
\beta_{j}:=-\frac{\alpha_{j}g^+_{10}(\alpha_{j}^{-1})}
{a_{j}(g^+_{11})'(\alpha_{j}^{-1})}, j=1,\ldots, t+r
$$

\item If all the $\beta_{j}$ are non-zero, add the resulting error
to a list of potential errors. 

\end{itemize}

\end{itemize}

\end{enumerate}

Note that error vectors added to the list in the above flow must have
the same syndrome as the received word, as follows from the
following proposition.\footnote{We thank M.~Twitto for pointing out
this observation.}

\begin{proposition}
Suppose that $\tl{\sigma}\in \efq[X]$ is separable, splits in
$\efq$, and satisfies $\tl{\sigma}(0)\neq 0$. Suppose also that
$\tl{\omega}\in \efq[X]$ satisfies 
$\deg(\tl{\omega})<\deg(\tl{\sigma})$ and $\tl{\omega}\equiv
\synd{y}\tl{\sigma}\mod (X^{d-1})$. Let 
$\tl{\bs{e}}$ be the vector with support
$\big\{\alpha^{-1}|\tl{\sigma}(\alpha)=0\big\}$ 
and corresponding non-zero entries obtained by Forney's formula
(\ref{eq:forney}) with $\tl{\sigma}$ and $\tl{\omega}$. 
Then $\synd{\tl{\bs{e}}}(X)=\synd{\bs{y}}(X)$.  
\end{proposition}

\begin{proof}
By dividing both $\tl{\sigma}$ and $\tl{\omega}$ by $\tl{\sigma}(0)$,
we may assume w.l.o.g.~that $\tl{\sigma}(0)=1$.  Note first that the
EEP related to $\tl{\bs{e}}$ is indeed $\tl{w}$: If 
$\hat{w}$ is the EEP related to $\tl{\bs{e}}$, then by Forney's
formula (\ref{eq:forney}), $\tl{\omega}-\hat{\omega}$ has
$\deg(\tl{\sigma})$ roots. By the degree assumption in the proposition,
$\deg(\tl{\omega}-\hat{\omega})<\deg(\tl{\sigma})$, which implies that
$\tl{\omega}-\hat{\omega}$ is the zero polynomial.

Writing ``$\equiv$'' for congruence modulo $(X^{d-1})$, it holds that
$$
\synd{\tl{\bs{e}}}\tl{\sigma}\equiv \tl{\omega} \equiv
\synd{y}\tl{\sigma}, 
$$
where the first congruence follows from the key equation, while the
second congruence holds by assumption. Hence
$X^{d-1}|(\synd{\tl{\bs{e}}}-\synd{y})$. As
$\deg(\synd{\tl{\bs{e}}}-\synd{y})\leq d-2$, this completes the
proof. 
\end{proof}

\section{Conclusions and open questions} \label{sec:conclusion}
We presented a conceptually simple fast Chase decoding algorithm for
RS codes, building on the theory of \groebner{} bases for
$\efq[X]$-modules. Working with ``the right'' minimization
problem in an $\efq[X]$-module results in a considerably simplified
polynomial-update algorithm, which is also  automatically suited to
working with vectors of evaluations. Our algorithms are not tied to
the BM algorithm for HD initialization, and practically any
syndrome-based HD decoding algorithm can be used for this purpose.  

It should be noted that both Algorithm A and B can be easily converted
to a fast GMD algorithm, by simply omitting the derivative
iteration. For Algorithm B, this means that there is no longer a need
to maintain the vectors of evaluations of the derivatives. Moreover,
a fast application of combinations of GMD and Chase decoding can be
obtained in this way.

We conclude with some open questions:

\begin{itemize}


\item Any Chase decoding algorithm for GRS codes is automatically also
a Chase decoding algorithm for their subfield-subcodes, the {\it
alternant} codes, which include BCH codes as a special case. However,
in \cite{Wu12}, the polynomial-update algorithm for binary BCH codes is
simpler than that of the corresponding RS codes.  Is there a way to
further simplify Algorithm A of the current paper for the case of
binary BCH codes?\footnote{We note that in a 
companion work \cite{SK1}, using a completely different method, some
of the authors have devised a syndrome-based Chase decoding algorithm
for binary BCH codes that is both conceptually simple and updates
polynomials of a lower degree than those of Algorithm 5 of
\cite{Wu12}. However, for completeness, it is still an interesting
question whether the current algorithm can be further simplified in
the case of binary BCH codes.} 

\item Interestingly, Algorithms A and B remain valid also when the
total number of errors is $\geq d-1$, as long as the conditions of
Theorem \ref{thm:main} are satisfied.  Can this be of any practical
value in some cases?  Note that while the output list size grows
exponentially beyond $d-1$, this can be solved by adding a small
number of CRC bits, or even without CRC bits, when the RS code is part
of a {\it generalized concatenated code} \cite[Sec.~18.8.2]{MS}.

\item Is it possible to further reduce the complexity by
using fast algorithms for basis reduction of polynomial matrices,
e.g., as in \cite{N16}, \cite{RS21} and the references therein, or 
by using fast algorithms for structured linear algebra, e.g., as in 
\cite{CJNSV15}? 

\end{itemize}

\appendix

\section*{Appendix}

\section{The case of enough correct modifications}\label{app:indirect}
In this appendix, we consider the case mentioned near the end of
Section \ref{sec:basic}, that is, the case where, although
the tested error pattern is not a direct hit, the difference between
the number of correct indices and incorrect indices is at least
$\veps-t$. The main result is Proposition \ref{prop:indirect}, which
shows that in this case, the outputs of Algorithm A can still be used
for finding the correct transmitted codeword.

We begin with a remark that will be useful in the proof of Proposition
\ref{prop:indirect}.
\begin{remark}\label{rem:inclusion}
{\rm 
For distinct $\alpha_1,\ldots,\alpha_{r+1}\in \efq^*$ and for
$\beta_1,\ldots,\beta_{r+1}\in \efq^*$, let
$$
M_{r+\frac{1}{2}}:=
M_r(\synd{y},\alpha_1,\ldots,\alpha_{r},\beta_1,\ldots,\beta_{r}) \cap
\big\{(u,v)|v(\alpha_{r+1}^{-1})=0\big\}.
$$
Taking $v(X):=\prod_{i=1}^r (1-\alpha_iX)$, and letting
$u_0:=\synd{y}\cdot v$, there exists a polynomial $f(X)\in \efq[X]$
such that, setting $u_f:=u_0+X^{d-1}\cdot f$, the $r$ ``derivative
equations'' from part 2 of Definition \ref{def:mr} are satisfied for
$(u_f,v)$ (this is just an interpolation problem for $f$, and it
obviously has a solution for $f$ of high enough degree). For such a
choice of $f$, clearly $(u_f,v)\in M_r\smallsetminus
M_{r+\frac{1}{2}}$.  

Similarly, taking now $v(X):=\prod_{i=1}^{r+1} (1-\alpha_iX)$, and
letting again $u_0:=\synd{y}\cdot v$, there exists a polynomial
$f(X)\in \efq[X]$ such that, setting $u_f:=u_0+X^{d-1}\cdot f$, the $r$
equations from part 2 of Definition \ref{def:mr} are satisfied for
$(u_f,v)$, while $\beta_{r+1}a_{r+1}v'(\alpha_{r+1}^{-1})\neq
\alpha_{r+1}u_f(\alpha_{r+1}^{-1})$ (again, this is an interpolation
problem for $f$, now with a lot of freedom in the choice of
$f(\alpha_{r+1}^{-1})$). For such a choice of $f$, clearly $(u_f,v) \in
M_{r+\frac{1}{2}}\smallsetminus M_{r+1}$.

We conclude that 
\begin{equation}\label{eq:inclusion}
M_r\supsetneq M_{r+\frac{1}{2}}\supsetneq M_{r+1}.
\end{equation}
}
\end{remark}

\begin{proposition}\label{prop:indirect}
Consider a vertex $\bs{v}=\big((\alpha_{i_{1}},\beta_{i_{1}}),\ldots,
(\alpha_{i_r},\beta_{i_r})\big)$ of $T$. Let 
$$
S:=\{\ell\in\{1,\ldots,r\}|\text{the error value at $\alpha_{i_\ell}$
is not $\beta_{i_{\ell}}$}\},
$$
and let
$$
S_1:=\{\ell\in\{1,\ldots,r\}|\text{$\alpha_{i_\ell}$ is not an error
location}\}\subseteq S.
$$
Finally, let $S_2:=S\smallsetminus S_1$. For any vertex $\bs{u}$ of
the tree $T$, let $\big\{\bgplus_0(\bs{u}),\bgplus_1(\bs{u})\big\}$ be
the \groebner{} basis calculated inductively from the root to the vertex
$\bs{u}$ by applying Algorithm A on the edges.  
Then if $r-|S|\geq \veps -t+|S_1|$, then it holds that 
$$
\lm\big(\bgplus_1(\bs{v})\big) < \lm\big(\bgplus_0(\bs{v})\big),
$$
and
\begin{equation}\label{eq:indirect}
\bgplus_1(\bs{v})=c\cdot(\omega,\sigma)\cdot \prod_{\ell\in S_2}
(X-\alpha_{i_{\ell}}^{-1})\cdot \prod_{\ell\in
S_1}(X-\alpha_{i_{\ell}}^{-1})^2 
\end{equation}
for some $c\in \efq^*$.
Hence, writing $\bgplus_1(\bs{v})=(\gplus_{10},\gplus_{11})$,
\begin{equation}\label{eq:inddiv}
\frac{g^+_{11}(X)}{\prod_{\ell=1}^r (X-\alpha_{i_{\ell}}^{-1})} =
\frac{c\cdot\sigma(X)}{\prod_{\ell\notin S} (X-\alpha_{i_{\ell}}^{-1})}
\cdot \prod_{\ell\in S_1}(X-\alpha_{i_{\ell}}^{-1}
).
\end{equation} 
\end{proposition}

\begin{remarknn} 
{\rm
Equation (\ref{eq:inddiv}) means that $g^+_{11}(X)/\prod_{\ell=1}^r
(X-\alpha_{i_{\ell}}^{-1})$ is an ``effective ELP'' corresponding to the
modification sequence in $\bs{v}$: Correct modifications are
canceled out from $\sigma$, wrong modification at correct locations
have no effect, while modification at locations without errors
effectively add error locations. 
}
\end{remarknn}

\begin{proof}[Proof of Proposition \ref{prop:indirect}]
Observe that $|S_1|$ is the number of wrongly-modified correct
coordinates for   $\bs{v}$, while $r-|S|$ is the number of
correctly-modified erroneous coordinates. Write $\delta:=|S|$,
$\delta_1:=|S_1|$. Modifying the order of the pairs defining $\bs{v}$
does not change the corresponding module
$$
M_{r}(\synd{y},\alpha_{i_1},\ldots,
\alpha_{i_r},\beta_{i_1},\ldots,\beta_{i_r}), 
$$
and hence also does not change the unique minimal element (by
Proposition \ref{prop:unique}). We may therefore assume  
w.l.o.g.~that $S=\{r-\delta+1,\ldots,r\}$, and that
$S_1=\{r-\delta_1+1,\ldots,r\}$. Hence, if 
$\ell$ is one of the $\delta-\delta_1$ smallest elements of $S$, then
$\alpha_{i_{\ell}}$ is an error location and $\beta_{i_{\ell}}$ is not
the corresponding error value. Similarly, if $\ell$ is one of
the $\delta_1$ largest elements of $S$, then $\alpha_{i_{\ell}}$ is
not an error location.

The idea of the proof is to trace the updates in Algorithm A, and
(loosely speaking) to show that for $\ell\in S_2$, $\bgplus_1$ is
multiplied once by $(X-\alpha_{i_{\ell}}^{-1})$, while for $\ell\in
S_1$, $\bgplus_1$ is multiplied twice by
$(X-\alpha_{i_{\ell}}^{-1})$. 

By assumption, the first
$r-\delta$ pairs in $\bs{v}$ are correct error locations and
corresponding values.  As we also assume that $r-\delta-\delta_1\geq
\veps-t$, it follows from Theorem \ref{thm:main} that when moving from
the root of $T$ to the vertex
$\bs{v}':=\big((\alpha_{i_{1}},\beta_{i_{1}}),\ldots,
(\alpha_{i_{r-\delta-\delta_1}},\beta_{i_{r-\delta-\delta_1)}})\big)$
at depth $r-\delta-\delta_1$, we have $\bgplus_1(\bs{v}')=c\cdot
(\omega,\sigma)$ for some $c\neq 0$. Moreover,
\begin{equation}\label{eq:vtag}
\lm(\bgplus_1(\bs{v}'))<\lm(\bgplus_0(\bs{v}')).
\end{equation}    

Now, for the next $\delta_1$ edges on the path from $\bs{v}'$ to
$\bs{v}$, we still have correctly-modified coordinates. Hence, in
Algorithm A, $\Delta_1=0$ for both the root and derivative iterations, and only
$\bs{g}_0$ might be modified in all of the corresponding $2\delta_1$ root
and derivative iterations. Moreover, $\bs{g}_0$ is indeed modified in each
and every one of the iterations, for otherwise the \groebner{} basis would
be unchanged, and hence the generated module would be
unchanged, contradicting (\ref{eq:inclusion}). 
Hence, writing $\bs{v}'':=\big((\alpha_{i_{1}},\beta_{i_{1}}),\ldots,
(\alpha_{i_{r-\delta}},\beta_{i_{r-\delta}})\big)$, we have
\begin{equation}\label{eq:vtt}
\big(\bgplus_0(\bs{v}''),\bgplus_1(\bs{v}'')\big) =
\Big(\bgplus_0(\bs{v}')\cdot\prod_{\ell=r-\delta-\delta_1+1}^{r-\delta}
(X-\alpha_{i_{\ell}}^{-1})^2, \bgplus_1(\bs{v}')\Big).
\end{equation}

It is now left to consider the last $\delta$ applications of Algorithm
A, on the path from $\bs{v}''$ to $\bs{v}$. 
Write
$\bs{v}_{r-\delta}=\bs{v}'',\bs{v}_{r-\delta+1},
\ldots,\bs{v}_{r}=\bs{v}$ for the consecutive vertices on the path
from $\bs{v}''$ to $\bs{v}$. We first prove by induction that for all
$\ell'\in \{r-\delta,\ldots,r-\delta_1\}$, 
\begin{equation}\label{eq:stwoa}
\lm(\bgplus_1(\bs{v}_{\ell'})) < \lm(\bgplus_0(\bs{v}_{\ell'})),
\end{equation}
and
\begin{eqnarray}
\bgplus_1(\bs{v}_{\ell'}) & = & \bgplus_1(\bs{v}') \cdot \prod_{\ell= 
r-\delta+1}^{\ell'} (X-\alpha_{i_{\ell}}^{-1})\nonumber\\
 & = & c\cdot(\omega,\sigma)\cdot \prod_{\ell= r-\delta+1}^{\ell'}
(X-\alpha_{i_{\ell}}^{-1}). \label{eq:stwob} 
\end{eqnarray}

The basis of induction, for $\ell'=r-\delta$ (where the product on
the  right of (\ref{eq:stwob}) is empty), follows from (\ref{eq:vtt})
and (\ref{eq:vtag}). For the step, assume 
that $\ell'\in\{r-\delta+1,\ldots,r-\delta_1\}$, and that
(\ref{eq:stwoa}), (\ref{eq:stwob}) hold for $\ell'-1$. As
$\alpha_{i_{\ell'}}$ is an error location, it follows from the
induction hypothesis that in the root iteration of Algorithm A,
$\Delta_1=0$, and consequently, $\bgplus_1=\bs{g}_1$. 

We claim that in the derivative iteration, $\Delta_1\neq 0$. For this,
let $\beta$ be the correct error value for the (correct) error
location $\alpha_{i_{\ell'}}$. Write  
$$
M:=M_{r-\delta+1}(\synd{y},\alpha_{i_1},\ldots,
\alpha_{i_{r-\delta}},\underline{\alpha_{i_{\ell'}}},\beta_{i_1},\ldots,
\beta_{i_{r-\delta}},\underline{\beta}).
$$  
Then clearly $(\omega,\sigma)\in M$, and since by the induction
hypothesis $\bgplus_1(\bs{v}_{\ell'-1})$ is obtained by multiplying
$(\omega,\sigma)$ by  a scalar polynomial,
$\bgplus_1(\bs{v}_{\ell'-1})$ is also in the module $M$. Hence
\begin{equation}\label{eq:fichs}
\beta \cdot a_{i_{\ell'}}
[g_{11}^+(\bs{v}_{\ell'-1})]'(\alpha_{i_{\ell'}}^{-1}) +
\alpha_{i_{\ell'}}[g_{10}^+(\bs{v}_{\ell'-1})](\alpha_{i_{\ell'}}^{-1})=0,
\end{equation}
where for $i\in\{0,1\}$ and for a vertex $\bs{u}$ of
$T$, we write
$\bgplus_i(\bs{u})=(g_{i0}^+(\bs{u}),g_{i1}^+(\bs{u}))$. 
Since it can be verified by the induction hypothesis that
$[g_{11}^+(\bs{v}_{\ell'-1})]'(\alpha_{i_{\ell'}}^{-1})\neq
0$,\footnote{Note that the induction hypothesis implies that
$g_{11}^+(\bs{v}_{\ell'-1})=f(X)\cdot\sigma(X)$ for some $f$ with
$f(\alpha_{i_{\ell'}}^{-1})\neq 0$.} it follows that replacing
$\beta$ by $\beta_{i_{\ell'}}\neq \beta$ on the left-hand side of
(\ref{eq:fichs}) will result in a non-zero value. This completes the
proof that $\Delta_1\neq 0$ on the derivative iteration, and hence,
using the induction  hypothesis for (\ref{eq:stwoa}),
(\ref{eq:stwob}), proves (\ref{eq:stwob}) for the induction step. 

For (\ref{eq:stwoa}), note that since in the root iteration
$\bgplus_1=\bgplus_1(\bs{v}_{\ell'-1})$ (that is,
$\bgplus_1=\bs{g}_1$), it follows from (\ref{eq:inclusion}) that

\begin{equation}\label{eq:bgplus}
\bgplus_0=\bgplus_0(\bs{v}_{\ell'-1})\cdot(X-\alpha_{i_{\ell'}}^{-1}).
\end{equation}

Hence, the induction hypothesis implies that $\lm(\bgplus_0)>X
\lm(\bgplus_1)$, and therefore after the derivative iteration 
it necessarily holds that $\lm(\bgplus_0)>\lm(\bgplus_1)$. This
completes the induction step for (\ref{eq:stwoa}).

Using (\ref{eq:vtt}) and (\ref{eq:bgplus}), and noting that by the
above it holds that $\lm(\bgplus_0)=\lm(\bs{g}_0)$ in the derivative
iteration, it also follows by induction that for all
$\ell'\in\{r-\delta+1,\ldots,r-\delta_1\}$, 

\begin{equation}\label{eq:a}
\lm\big(\bgplus_0(\bs{v}_{\ell'})\big) = \lm\Big(\bgplus_0(\bs{v}') \cdot \prod_{\ell= 
r-\delta-\delta_1+1}^{r-\delta} (X-\alpha_{i_{\ell}}^{-1})^2 \cdot
\prod_{\ell=  r-\delta+1}^{\ell'} (X-\alpha_{i_{\ell}}^{-1})\Big),
\end{equation} 
where we have used $\lm (f\cdot\bs{h}) = X^{\deg(f)}\cdot
\lm(\bs{h})$ for $f\in \efq[X]$ and $\bs{h}\in \efq[X]^2$.

To complete the proof, we will prove by induction that for all
$\ell'\in \{r-\delta_1,\ldots,r\}$, $\lm(\bgplus_1(\bs{v}_{\ell'})) <
\lm(\bgplus_0(\bs{v}_{\ell'}))$ and
\begin{equation}\label{eq:soneb}
\bgplus_1(\bs{v}_{\ell'})=c\cdot(\omega,\sigma) \cdot
\prod_{\ell=r-\delta+1}^{r-\delta_1} (X-\alpha_{i_{\ell}}^{-1}) \cdot 
\prod_{\ell=r-\delta_1+1}^{\ell'}(X-\alpha_{i_{\ell}}^{-1})^2 . 
\end{equation} 
The basis of the induction, for $\ell'=r-\delta_1$, follows from
(\ref{eq:stwoa}), (\ref{eq:stwob}). 

To continue, recall that in both the root and the derivative iterations
of algorithm A, if the leading monomial $\bs{m}$ of one of the pairs
is changed, then it is changed to $X\bs{m}$.
Hence, it follows from substituting $\ell'=r-\delta_1$ in
(\ref{eq:stwob}), (\ref{eq:a}) and the fact that 
$\lm(\bgplus_1(\bs{v}'))<\lm(\bgplus_0(\bs{v}'))$,
that for all $\ell'\in\{r-\delta_1+1,\ldots, r\}$, it holds that
$\lm(\bgplus_1)<\lm(\bgplus_0)$ for both the root and derivative
iterations of Algorithm A.\footnote{In detail, note that
$\lm(\bgplus_0(\bs{v}_{r-\delta_1})) =
X^{2\delta_1+\delta-\delta_1}\cdot\lm(\bgplus_0(\bs{v}'))$, 
while
$\lm(\bgplus_1(\bs{v}_{r-\delta_1})) =
X^{\delta-\delta_1}\cdot\lm(\bgplus_1(\bs{v}'))$. 
Hence, for all $\ell'\in\{r-\delta_1+1,\ldots,r\}$, we have 
\begin{eqnarray*}
\lm(\bgplus_1(\bs{v}_{\ell'})) & \leq &
X^{2\big(\ell'-(r-\delta_1)\big)}
X^{\delta-\delta_1}\lm(\bgplus_1(\bs{v}')) 
\\ & \leq & X^{\delta+\delta_1} \lm(\bgplus_1(\bs{v}')) 
\text{ (substituting $\ell'=r$)}
\\ & < & \lm(\bgplus_0(\bs{v}_{r-\delta_1})) \leq
\lm(\bgplus_0(\bs{v}_{\ell'})) . 
\end{eqnarray*}
}

Hence, for $\ell'\in \{r-\delta_1,\ldots,r-1\}$ there are only three
possible ways in which $g_1^+(\bs{v}_{\ell'})$ can be updated 
to $g_1^+(\bs{v}_{\ell'+1})$: (1)
$g_1^+(\bs{v}_{\ell'+1})=g_1^+(\bs{v}_{\ell'})$, (2)
$g_1^+(\bs{v}_{\ell'+1})=g_1^+(\bs{v}_{\ell'})\cdot
(X-\alpha_{i_{\ell'+1}}^{-1})$, or (3)
$g_1^+(\bs{v}_{\ell'+1})=g_1^+(\bs{v}_{\ell'})\cdot
(X-\alpha_{i_{\ell'+1}}^{-1})^2$.   

For the induction step of the proof of (\ref{eq:soneb}), assume that 
$\ell'\in\{r-\delta_1,\ldots,r-1\}$, and that    
(\ref{eq:soneb}) holds for $\ell'$. Considering options (1)--(3) 
above, it is sufficient to prove that when moving from
$\bs{v}_{\ell'}$ to $\bs{v}_{\ell'+1}$, it holds that $\Delta_1\neq 0$
for both the root and derivative iterations of Algorithm A. 

As $\alpha_{i_{\ell'+1}}$ is not an error location, it follows
from the induction hypothesis that 
$\alpha_{i_{\ell'+1}}$ is not a root of $g_{11}^+(\bs{v}_{\ell'})$, and
therefore $\Delta_1\neq 0$ in the root
iteration. Hence, at the end of the root iteration, we have
\begin{equation}\label{eq:mid}
\bgplus_1=\bgplus_1(\bs{v}_{\ell'})\cdot (X-\alpha_{i_{\ell'+1}}^{-1}).
\end{equation}
Therefore,
$$
[g_{11}^+]'(\alpha_{i_{\ell'+1}}^{-1})=
[g_{11}^+(\bs{v}_{\ell'})](\alpha_{i_{\ell'+1}}^{-1}) \neq 0,
$$
where the last inequality follows again from the induction
hypothesis. Also, it follows from (\ref{eq:mid}) that
$g_{10}^+(\alpha_{i_{\ell'+1}}^{-1})=0$, and finally that $\Delta_1\neq
0$ in the derivative iteration, as required.
\end{proof}

\section{Simplifications for Algorithm A}\label{app:simplifications}

\subsection{Moving from two pairs of polynomials to two
polynomials}\label{app:pairs} 
In Algorithm A, two \emph{pairs} of polynomials  
have to be maintained, rather than just two polynomials. In the above
form, the algorithm will work even if $\veps\geq 2t$, where $\veps$ is
the total number of errors. However, as we shall now see, if
$\veps\leq 2t-1$, then there is no need to maintain 
the first coordinate of the \groebner{} basis. 

In order to omit the first entry in each pair, we
have to consider the following questions:
\begin{enumerate}

\item How can we efficiently calculate $g_{j0}(\alpha_r^{-1})$
($j\in\{0,1\}$) when only $g_{j1}$ is available?

\item How can we find $\lm(\bs{g}_0)$ without maintaining $g_{00}$
(recall that the leading monomial of $\bs{g}_0$ is on the left)?

\end{enumerate} 

The answer to the second question is almost trivial: Introduce a
variable $d_0$ to track the degree of $g_{00}$. Whenever $j^*=0$,
increase $d_0$ by $1$, and in all other cases keep $d_0$
unchanged (note that when $0\in J$ but $0\neq j^*$, 
$\lm(\bgplus_0)=\lm(\bs{g}_0)$, which justifies keeping $d_0$
unchanged). Now $\lm(\bs{g}_0)=(X^{d_0},0)$. 

So, let us turn to the first question. We know that for all $r$ and
all $(u,v)\in
M_r(\synd{y},\alpha_1,\ldots,\alpha_r,\beta_1,\ldots,\beta_r)$, we
have $u\equiv \synd{y}v \mod (X^{2t})$, and hence one can calculate
$u(\alpha_r^{-1})$ directly from $v$ if $\deg(u)\leq 2t-1$ (see
ahead). So, our first task is to verify that if $\veps\leq 2t-1$ 
(so that $r\leq 2t-1-t=t-1$), we have $\deg(g_{10})\leq 2t-1$ and
$\deg(g_{20})\leq 2t-1$ for all \koetter{}'s iterations involved in fast
Chase decoding, assuming the hypotheses of Theorem \ref{thm:main}
hold.  

We will first need a small modification of the first part of
\cite[Prop.~2]{BHNW13}. To 
keep this paper self-contained, we will also include the proof.  From
this point on, we will say that a monomial in $\efq[X]^2$ is {\it on the
left} if it contains the unit vector $(1,0)$, and {\it on the right}
if it contains the unit vector $(0,1)$. 

\begin{proposition}[\cite{BHNW13}]\label{prop:degrees}
Let $\{\bs{h}_0=(h_{00},h_{01}),\bs{h}_1=(h_{10},h_{11})\}$ be a
\groebner{} basis for $M_0$ with respect to the monomial ordering $<$,
and suppose that the leading monomial of $\bs{h}_0$ is on the left, while
the leading monomial of $\bs{h}_1$ is on the right. Then
$\deg(h_{00}(X))+\deg(h_{11}(X)) = 2t$. 
\end{proposition}
\begin{proof}
Since $(\synd{y},1)$ is in the $\efq[X]$-span of
$\{\bs{h}_0,\bs{h}_1\}$, it follows that $1\in (h_{01},h_{11})$, and
hence that $h_{01}$ and $h_{11}$ are relatively prime. Now suppose
that $\alpha(X),\beta(X)\in \efq[X]$ are such that 
$\alpha(X)\bs{h}_0-\beta(X)\bs{h}_1=(\gamma(X),0)$ for some
$\gamma(X)$. Then $\alpha(X)h_{01}(X)=\beta(X)h_{11}(X)$, and because
$\gcd(h_{01},h_{11})=1$, this implies that $h_{11}(X)|\alpha(X)$,
$h_{01}(X)|\beta(X)$, 
$$
\frac{\alpha(X)}{h_{11}(X)}=\frac{\beta(X)}{h_{01}(X)}
$$
and these two equal rational functions are in fact a polynomial in
$\efq[X]$. Write $r(X)\in \efq[X]$ for this polynomial. Let
$\pi_0\colon \efq[X]^2\to\efq[X]$ be the projection to the first coordinate.
Now, the second coordinate of the vector
$$
\bs{f}:=h_{11}(X)\bs{h}_0-h_{01}(X)\bs{h}_1\in M_0
$$
is $0$. Also, for any $\alpha(X),\beta(X)$ as above, it follows from
the definition of $r(X)$ that
$$
\alpha(X)\bs{h}_{0}-\beta(X)\bs{h}_{1} = r(X)\cdot \bs{f}.
$$

This shows that $\pi_0(\bs{f})$ has the lowest degree in $\pi_0\big(M_0\cap
(\efq[X]\times\{0\})\big)$. Now, as $M_0$ is generated as an $\efq[X]$-module by
$\{(X^{2t},0),(\synd{y}(X),1)\}$, we know that this lowest degree is
$2t$. Hence $\deg(\pi_0(\bs{f}))=2t$. Now,
\begin{eqnarray*}
\deg(\pi_0(\bs{f})) & = & \deg\big(h_{11}(X)h_{00}(X)-h_{01}(X)h_{10}(X)\big)\\
           & = & \deg\big(h_{11}(X)h_{00}(X)\big),
\end{eqnarray*}
because by assumption $\deg(h_{11})\geq \deg(h_{10})+1$ and
$\deg(h_{00})> \deg(h_{01})-1$, so that
$\deg(h_{11}h_{00})>\deg(h_{01}h_{10})$. 
\end{proof}

With Proposition \ref{prop:degrees}, we can now prove that for all
iterations of \koetter{}'s algorithm,
$\deg(g_{10})\leq 2t-1$ and $\deg(g_{20})\leq 2t-1$ when
$\veps\leq 2t-1$. Before the proof, it will be useful to introduce
some additional notation.

\begin{definition}\label{def:itau}
{\rm
For $i=1,\ldots,r$, $j\in\{0,1\}$, and ${\tau}\in\{\broot,\bder\}$ write
$\bs{g}_j(i;{\tau})=(g_{j0}(i;{\tau}),g_{j1}(i;{\tau}))$ and
$\bgplus_j(i;{\tau})=(\gplus_{j0}(i;{\tau}),\gplus_{j1}(i;{\tau}))$ 
for the values in the root iteration (${\tau}=\broot$) or the derivative iteration
(${\tau}=\bder$) of Algorithm A corresponding to adjoining error location
$\alpha_i$. By convention, $\{\bs{g}_0(1;\broot),\bs{g}_1(1;\broot)\}$
is a  \groebner{} basis for $M_0$ with $\lm(\bs{g}_0(1;\broot))$ on the
left and $\lm(\bs{g}_1(1;\broot))$ on the right. Note that for all
$i$, $\bs{g}_j(i,\bder)=\bgplus_j(i,\broot)$ ($j=1,2$), and for all 
$i\geq 2$, $\bs{g}_j(i,\broot)=\bgplus_{j}(i-1,\bder)$ ($j=1,2$).
}
\end{definition}

\begin{proposition}\label{prop:okmod}
Suppose that the condition in part 2 of Theorem \ref{thm:main} holds,
and that the total number $\veps$ of errors is exactly $t+r$.
Then for all $i\in\{1,\ldots,r\}$, all $j\in\{0,1\}$ and all
${\tau}\in\{\broot,\bder\}$, $\deg(g^+_{j0}(i;{\tau}))\leq \veps$ and 
$\deg(g^+_{j1}(i;{\tau}))\leq \veps$. 
\end{proposition}

\begin{proof}
By Theorem \ref{thm:main}, $(\omega,\sigma)=c\cdot \bgplus_1(r;\bder)$
for some $c\in \efq^*$ (as the leading monomial of $(\omega,\sigma)$
is on the right). Note that for all $i$, $j$, and ${\tau}$, we have
$\lm(\bgplus_j(i;{\tau}))\geq \lm(\bs{g}_j(i;{\tau}))$, and so for all
$i$ and ${\tau}$, we must have $\lm(\bs{g}_1(i;{\tau}))\leq
\lm(\omega,\sigma)=(0,X^{\veps})$. In particular,
$\deg(g_{11}(i;{\tau}))\leq \veps$ and $\deg(g_{10}(i;{\tau}))\leq
\veps-1$. The same argument applies also to $g^+_{10}(i;{\tau})$ and  
$g^+_{11}(i;{\tau})$. 

Turning to $\bs{g}_{0}(i;{\tau})$, note that for all $i$ and $\tau$,
$\lm(\bgplus_j(i;{\tau}))>\lm(\bs{g}_j(i;{\tau}))$ for at most one
$j\in\{0,1\}$. Also, for $j\in\{0,1\}$ and for each $i$ and ${\tau}$
with $\lm(\bgplus_j(i;{\tau}))>\lm(\bs{g}_j(i;{\tau}))$, we have
$\lm(\bgplus_j(i;{\tau}))=X\lm(\bs{g}_j(i;{\tau}))$. Since the degree of the second
coordinate of $\bgplus_1(i;{\tau})$ (the coordinate containing the leading
monomial) must increase from $\deg(g_{11}(1;\broot))$ for $i=1$ and
${\tau}=\broot$ to $\deg(\sigma)=\veps$ for $i=r$ and ${\tau}=\bder$, we see that
$$
\big|\big\{(i,{\tau})|
\lm(\bgplus_1(i;{\tau}))>\lm(\bs{g}_1(i;{\tau}))\big\}\big| =
\veps-\deg(g_{11}(1;\broot)), 
$$
and therefore,\footnote{Actually, by
(\ref{eq:inclusion}) we can replace ``$\leq$'' by ``$=$'' in the
following equation.}
\begin{eqnarray*}
\big|\big\{(i,{\tau})|
\lm(\bgplus_0(i;{\tau}))>\lm(\bs{g}_0(i;{\tau}))\big\}\big| & \leq & 
2r-(\veps-\deg(g_{11}(1;\broot)))\\
&=& \deg(g_{11}(1;\broot)) +r-t.
\end{eqnarray*}
Hence, for all $i$ and ${\tau}$,
\begin{eqnarray*}
\deg(g_{00}^+(i;{\tau})) &\leq & \deg(g_{00}^+(r;\bder)) \text{ ($\lm$
(on the left) does not decrease)}\\ 
 & \leq & \deg(g_{00}(1;\broot)) + \deg(g_{11}(1;\broot)) +r-t\\
 & = & t+r=\veps \text{ (by Proposition \ref{prop:degrees})}.
\end{eqnarray*}

Finally, since the leading monomial of $\bgplus_0(i;{\tau})$ is on the
left, we must have
$\deg(g_{01}^+(i;{\tau}))-1<\deg(g_{00}^+(i;{\tau}))\leq 
\veps$, which proves that $\deg(g_{01}^+(i;{\tau}))\leq \veps$.
\end{proof}

Using Proposition \ref{prop:okmod}, we can calculate
$g_{j0}(\alpha_r^{-1})$ in Algorithm A while maintaining only the right
polynomials $g_{j1}$ ($j\in\{0,1\}$). We shall now describe an
efficient $O(t)$ method for calculating this evaluation.

For a polynomial $v(X)\in \efq[X]$, assume that $\delta:=\deg(v)\leq
\veps\leq 2t-1$, and write $v(X)=v_0+v_1X+\cdots+v_{2t-1}X^{2t-1}$. For
short, write
$S(X)=S_0+S_1X+\cdots+S_{2t-1}X^{2t-1}:=\synd{y}(X)$. Then for
$\beta\in \efq$, $\big(Sv\bmod (X^{2t})\big)(\beta)$ can be expressed 
as 
\begin{eqnarray}
&&S_0v_0  + \nonumber\\
&&(S_0v_1+S_1v_0)\beta + \nonumber\\
&&(S_0v_2+S_1v_1+S_2v_0)\beta^2 + \nonumber\\
&&\vdots  \nonumber\\
&&(S_0v_{2t-1}+S_1v_{2t-2}+S_2v_{2t-3}+\cdots+ S_{2t-1}v_0)\beta^{2t-1}.
\label{eq:mod} 
\end{eqnarray}

For $j\in\{0,\ldots,2t-1\}$, let $A_j(v,\beta)$ be the sum over the $j$-th
column of (\ref{eq:mod}). Then
$$
A_j(v,\beta)=S_j\beta^j(v_0+v_1\beta+\cdots+v_{2t-1-j}\beta^{2t-1-j}).
$$
If $2t-1-j\geq \delta(=\deg(v))$, then
$A_j(v,\beta)=S_j\beta^j v(\beta)$. Hence, if $v(\beta)=0$ (which we
will assume from this point on, considering the previous root iteration of 
Algorithm A), then 
\begin{equation}\label{eq:modsum}
\big(Sv\bmod (X^{2t})\big)(\beta) = 
\sum_{j=0}^{2t-1}A_j(v,\beta)
 =  \sum_{j=2t-\delta}^{2t-1}A_j(v,\beta).
\end{equation}
The sum on the right-hand side of (\ref{eq:modsum}) may be
calculated recursively. For this, let
$$
\tl{A}_j(v,\beta):=\beta^j\sum_{i=0}^{2t-1-j} v_i \beta^i,
$$
so that $A_j(v,\beta)=S_j\tl{A}_j(v,\beta)$. Then
$\tl{A}_{2t-\delta-1}=0$ (as $v(\beta)=0$), and for all 
$j\in\{2t-\delta-1,\ldots,2t-2\}$,
$\tl{A}_{j+1}(v,\beta)=\beta\tl{A}_j(v,\beta)-\beta^{2t}v_{2t-1-j}$,
that is, 
\begin{equation}\label{eq:recur}
\frac{\tl{A}_{j+1}(v,\beta)}{\beta^{2t}} =
\beta\cdot\frac{\tl{A}_j(v,\beta)}{\beta^{2t}}-v_{2t-1-j}.  
\end{equation}

Calculating $\beta^{2t}$ takes $O(\log_2(2t))$
squarings and multiplications. In fact, this can be calculated once,
before starting the depth-first search in the tree, for \emph{all}
non-reliable coordinates (not just for those corresponding to a
particular vertex). After that, each one of the $\delta$ iterations of
(\ref{eq:recur}) in the calculation of the sum (\ref{eq:modsum})
requires $2$ finite-field multiplications: one for moving from
$\tl{A}_j(v,\beta)/\beta^{2t}$ to $\tl{A}_{j+1}(v,\beta)/\beta^{2t}$,
and one for multiplying by $S_{j+1}$ before adding to an accumulated
sum.  Then, after the calculation of the accumulated sum, one additional
multiplication by $\beta^{2t}$ is required. We conclude that
calculating $\big(Sv\bmod (X^{2t})\big)(\beta)$ requires a total of
$2\delta+1$ finite-field multiplications (recall that
$\delta=\deg(v)$). 

For comparing the complexity with \cite[Alg.~1]{Wu12}, let us now
estimate the total number of finite-field multiplications 
required for performing the above variant of Algorithm A. For this
purpose, for $\tau\in\{\broot,\bder\}$, let $2\partial(r;\tau)$ be an
upper bound on the sum of the 
degrees of $g^+_{01}(r;\tau)$ and $g^+_{11}(r;\tau)$. 

The following proposition proves that we may take
$\partial(r;\broot)=t+r-1/2$ and $\partial(r;\bder)=t+r$.
\begin{proposition}\label{prop:partial}
For all edge connecting a vertex at depth $r-1$ to a vertex at depth
$r$,
$$
\deg(g^+_{00}(r;\tau)) + \deg(g^+_{11}(r;\tau)) \leq \begin{cases}
2t+2r-1 & \text{if }\tau=\broot\\
2t+2r & \text{if }\tau=\bder.
\end{cases}.
$$
Consequently, since $\deg(g^+_{01}(r;\tau))\leq
\deg(g^+_{00}(r;\tau))$ for all $\tau$,\footnote{Recall that the
leading monomial of $\bgplus_0(r;\tau)$ is on the left.} it
also holds that  
$$
\frac{1}{2}\Big(\deg(g^+_{01}(r;\tau)) + \deg(g^+_{11}(r;\tau))\Big)
\leq \begin{cases}
t+r-1/2 & \text{if }\tau=\broot\\
t+r & \text{if }\tau=\bder. 
\end{cases}
$$
\end{proposition}

\begin{proof}
By Proposition \ref{prop:degrees}, the assertion holds for $r=0$ (with an
obvious convention in this case). Now, for each root and derivative
iteration, the leading monomial increases for exactly one value of $j$
(namely $j=j^*$), for which it is multiplied by $X$. Since for all
$i,\tau$, the leading monomial of $\bgplus_0(i;\tau)$ is
$(X^{\deg(g^+_{00}(i;\tau))},0)$ while the leading monomial of
$\bgplus_1(i;\tau)$ is $(0,X^{\deg(g^+_{11}(i;\tau))})$, and since
there is a total of either $2r-1$ (or $2r$) root and derivative
iterations for $\tau=\broot$ (resp., $\tau=\bder$), the assertion
follows.  
\end{proof}

In the following complexity estimation for the number of
multiplications on an edge connecting a vertex at depth $r-1$ to a
vertex at depth $r$, we assume that all
involved discrepancies are non-zero. In the other cases, which are
typically rare, the complexity is lower. 

\begin{itemize}

\item {\bf In the root iteration:} 

\begin{itemize} 

\item {\bf Evaluation:} For $j=0,1$, we have to calculate
$g_{j1}(\alpha_r^{-1})$. Hence, we have
two substitutions in polynomials whose sum of degrees is at most
$2\partial(r-1;\bder)$, which requires a total of at most
$2\partial(r-1;\bder)$ multiplications.\footnote{Using {\it
Horner's method}.}  

\item {\bf Multiplication of a polynomial by a constant:} For
$j\neq j^*$, we have to calculate the constant $\Delta_j/\Delta_{j^*}$,
which requires a single multiplication (assuming that we have a table
for calculating inverses), and to multiply two polynomials whose sum
of  degrees is at most $2\partial(r-1;\bder)$ by a constant.\footnote{
Actually, in the current form of Algorithm A, we multiply \emph{the
same} polynomial $\gplus_{j^*1}(r-1;\bder)$ twice by a constant,
instead of multiplying the two polynomials $\gplus_{j1}(r-1;\bder),
\gplus_{j^*1}(r-1;\bder)$ whose sum of degrees was bounded in
Proposition \ref{prop:partial}. However, this can be resolved by
changing the update rule of Algorithm A for $j\neq j^*$ into
$\bgplus_j:=\frac{\Delta_{j^*}}{\Delta_j} \bs{g}_j-\bs{g}_{j^*}$. A
similar remark is relevant also for the complexity analysis for
Algorithm C ahead. \label{footnote:jstar}
}
This requires a total of
$1+2\partial(r-1;\bder)+2=2\partial(r-1;\bder)+3$ multiplications (the
``$+2$'' accounts for the fact that a polynomial of degree $d$ has
$d+1$ coefficients).   

\end{itemize}

\item {\bf In the derivative iteration:} 

\begin{itemize} 

\item {\bf Evaluation:} For $j=0,1$, we have to calculate
$g_{j1}'(\alpha_r^{-1})$. In general, this requires at most
$2\partial(r;\broot)-2$ multiplications (in characteristic $2$, only up
to $\partial(r;\broot)$ multiplications are required\footnote{In
characteristic $2$, for all polynomial $u(X)$, there exists a
polynomial $f(X)$ such that $u'(X)=f(X)^2$ with $\deg(f)\leq \lceil
\deg(u)/2\rceil-1$. Moreover, the coefficients of $f$ are obtained as the
square roots of the odd coefficients of $u$, and the square root
calculation amounts to a cyclic shift when elements are represented
according to a {\it normal basis} over $\eftwo$.}). We then have to
calculate $\beta_ra_r$ and multiply the two evaluation results by 
$\beta_ra_r$, which adds $3$ multiplications. Finally, we have to
calculate  
$g_{j0}(\alpha_r^{-1})$ (using the above method) and
multiply by a constant for $j=0,1$, requiring at most
$2(2\partial(r;\broot)+1+1)$ 
multiplications. Hence, the overall number of multiplications for
evaluation in the derivative step is $6\partial(r;\broot)+5$ in
general (or $5\partial(r;\broot)+7$ in characteristic $2$). 

\item {\bf Multiplication of a polynomial by a constant:}
This is the same as in the root iteration: a total of at most
$2\partial(r;\broot)+3$ multiplications. 

\end{itemize}

\end{itemize}

Summing up, we obtain that the total number of multiplications is at
most 
\begin{multline*}
M_{\mathrm{A}}:=2\partial(r-1;\bder) + 2\partial(r-1;\bder)+3 +
6\partial(r;\broot)+5 + 2\partial(r;\broot)+3\\ =
 4\partial(r-1;\bder) +
8\partial(r;\broot)+11= 2(2t+2r-2)+4(2t+2r-1)+11 
\\ =12t+12r+3
\end{multline*}
in general, while in characteristic $2$, the number of multiplications
is at most
\begin{multline*}
M_{\mathrm{A}}-\partial(r;\broot)+2 =12t+12r +3 - (t+r-1/2)+2
\\ =11t+11r+5.5.
\end{multline*}
Note that the fraction appears because this is
just a bound, but since the number of multiplications is an integer,
it is bounded by
$$
M'_A:=11t+11r+5 
$$
in characteristic $2$.

Next, we would like to make a similar calculation for
Wu's polynomial update algorithm, \cite[Alg.~1]{Wu12}. We note that the
complexities of Cases $3$--$8$ in Step 3 
of \cite[Alg.~1]{Wu12} are similar, and we will assume any one of these
cases, in analogy to the above assumption that no discrepancy is zero
in our algorithm.  To make concrete statements, we will focus on Case
$3$ as a representative for all of these cases. Similarly to the above,
we let $\partial$ be an upper bound on the degree for all involved
polynomials (before update) on an edge connecting a vertex at depth
$r-1$ to a vertex at depth $r$. While an
exact account of the polynomial degrees in Wu's algorithm is outside
the scope of the current paper, it seems 
reasonable to assume that we may take $\partial= t+r-1$ for Wu's
algorithm.

\begin{itemize}

\item {\bf Direct Evaluation in Step 2:} There are two evaluations of polynomials of
degree up to $\partial$ (in the calculations of $\bar{\Lambda}_i$ and
$\bar{\mathcal{B}}_i$), plus $3$ additional multiplications: one for
calculating $y_i\alpha_i$, and two for multiplying evaluation results
by $\alpha_i y_i$, resulting in a total of $2\partial+3$
multiplications.

\item{\bf Evaluation by the recursions \cite[Eq.~(23),(24)]{Wu12} in
Step 2:}
It seems that these recursions for 
calculating $\bar{\Omega}_i$, $\bar{\Theta}_i$ serve the same
purpose as a calculation explained above for Algorithm A: to
calculate $\big(Sv\bmod (X^{2t})\big)(\beta)$ for some polynomial $v$
and some element $\beta\in \efq^*$. We will therefore assume that each
of these calculations requires $2\partial+1$ multiplications, for a
total of $2(2\partial+1)=4\partial+2$ multiplications.

\item {\bf Additional multiplications in Step 2}:
In the two last lines of Step 2, there are $5$ additional
multiplications.  

\item {\bf Multiplication of a polynomial by a constant in Case 3 of
Step 3:} There are up
to $3(\partial+1)$ multiplications coming from 
multiplying a polynomial of degree up to $\partial$ by a scalar, plus $3$
additional multiplications (calculating $\alpha_i^{-1}\Psi_i$, and
multiplying its inverse by two constants) for calculating the relevant
scalars, resulting in $3\partial+6$ multiplications.

\item {\bf Syndrome update in Step 1:} There are $t+r$ syndrome entries to update,
each requiring a single multiplication. 

\end{itemize}

Summing-up, we obtain that the total number of multiplications in Wu's
algorithm on an edge between depth $r-1$ and depth $r$ is at most 
\begin{multline*}
M_{\mathrm{Wu}}:=2\partial+3 + 4\partial+2 + 5 + 3\partial+6 
+t+r\\ =9\partial+16+t+r \\=9(t+r-1)+16+t+r=10t+10r+7.
\end{multline*}
multiplications. Comparing $M_{\mathrm{Wu}}$ to $M_{\mathrm{A}}$ in
the general case and to $M'_{\mathrm{A}}$ in the case of
characteristic $2$, we see that Wu's algorithm has a somewhat lower
complexity, by a factor of about $5/6$ in general, or $10/11$ in
characteristic $2$. However, in Section \ref{app:lowdeg} we will
present yet an additional variant of Algorithm A (namely, Algorithm
C), that has a lower complexity than Wu's algorithm. 

It should be noted that the above complexity comparison does not
account for exhaustive root searches, as it is 
reasonable to assume that the probability of falsely meeting the
stopping criterion of Section \ref{app:stop} ahead is similar to the
corresponding probability for Wu's stopping criterion.

\subsection{A heuristic stopping criterion}\label{app:stop}
To reduce the number of required exhaustive root searches for the ELP
in Algorithm A, it is useful to introduce a heuristic stopping
criterion, which determines whether or not an exhaustive root search
is required. Such a stopping criterion must never miss the correct
ELP, but is allowed to falsely trigger an exhaustive root search with
a low probability. 

In \cite[Sec.~V]{Wu12}, Wu introduced such a heuristic criterion for his
algorithm, based on an LFSR-length variable. For Algorithm A, it is
possible to obtain a similar criterion based on the {\it discrepancy}
$\Delta_1$. Using the terminology of Section \ref{sec:tree}, suppose
that the total number of errors is $t+r$, and that there are $r+1$
errors on $I$, for some $r<\rmax$. Then by Theorem \ref{thm:main},
the correct EEP $\omega$ and ELP $\sigma$ will appear (up to a
multiplicative scaler) as the pair $\bgplus_1$ both for some vertex
$\bs{v}$ at depth $r+1$ and for its parent $\bs{u}$ at depth
$r$. 

Moreover, on the edge connecting $\bs{u}$ to $\bs{v}$, we must
have $\Delta_1=0$, both for the root iteration and for the derivative
iteration, by Forney's formula (\ref{eq:forney}). Hence, demanding that
$\Delta_1=0$ for both the root iteration and the derivative iteration
will never miss the true ELP under the above assumptions.

Special care should be taken for the case considered in Appendix
\ref{app:indirect}, as one can verify that $\Delta_1=0$ twice also if
a correct error location, $\alpha_{r+1}$, is encountered after the 
condition of Proposition \ref{prop:indirect} holds (we omit the
proof). While Proposition \ref{prop:indirect} can 
be used to restore the correct ELP also in 
such a case, this is outside the scope
of the current paper. Here, we will only specify a method to avoid a
useless exhaustive evaluation in these cases.\footnote{We thank A.~Dor
for pointing this out.} 

Observing (\ref{eq:indirect}), we see that for the case
considered in Proposition \ref{prop:indirect}, the estimated ELP and
its derivative have at least one common root in
$\{\alpha_1,\ldots,\alpha_r\}$. Hence, to avoid an unnecessary exhaustive
root search in such a case, one can first evaluate the
estimated ELP and its derivative on $\{\alpha_1,\ldots,\alpha_r\}$,
and then check that there are no common roots. Note that in case of a
direct hit, this condition does hold, as the ELP is separable.

To conclude, the stopping criterion now has the following form:
\begin{enumerate}

\item $\Delta_1=0$ for both the root an derivative iterations on the
edge connecting $\bs{u}$ and $\bs{v}$, and

\item The estimated ELP, $g_{11}(X)$, and its derivative have no common
roots on the $r$ locations corresponding to the vertex $\bs{u}$.

\end{enumerate}
Note that Condition 2 should be checked only if Condition 1 holds, and
hence rarely.
In other cases that there is no need to perform exhaustive evaluation,
it seems reasonable to heuristically assume that the
probability that $\Delta_1=0$ for both the root and the derivative
iterations is about $1/q^2$, and hence small.

\subsection{Working with low-degree polynomials}\label{app:lowdeg}
In this section, we will show that by using an appropriate
transformation, one can work with polynomials whose degrees 
typically grow from $0$ to $2\rmax$, instead of typically growing from
$t$ to $t+\rmax$ (see ahead for a detailed complexity comparison with
\cite[Alg.~1]{Wu12}). 

Until this point, we have used only the monomial ordering $<_{-1}$ of
\cite{Fitz95}. In this section, we will use the general definition of
Fitzpatrick's monomial ordering, which we shall now recall. Define a
monomial ordering $<_w$ on 
$\efq[X]^2$ as follows: $(X^{j_1},0)<_w(X^{j_2},0)$ iff $j_1<j_2$,
$(0,X^{j_1})<_w(0,X^{j_2})$ iff $j_1<j_2$, and
$(X^{j_1},0)<_w(0,X^{j_2})$ iff $j_1\leq j_2+w$.  Note again that this is a 
monomial ordering even when $w$ is not positive. We will write $\lm_w(u,v)$
for the leading monomial of $(u,v)$ with respect to $<_w$.

Let $\{\bh_0=(h_{00},h_{01}),\bh_1=(h_{10},h_{11})\}$ be a \groebner{}
basis for $M_0$ with respect 
to the monomial ordering $<_{-1}$ such that the leading monomial of
$\bh_0$ is on the left, while the leading monomial of $\bh_1$ is on
the right. Since $\{\bh_0,\bh_1\}$ is also a free-module basis, every
element $(u,v)\in M_0$ can be written as 
$(u,v)=f_0(X)\bh_0 + f_1(X)\bh_1$ for a unique pair $(f_0(X),f_1(X))$, and
the map
\begin{eqnarray*}
\mu\colon M_0 & \longrightarrow & \efq[X]^2\\
(u,v) & \longmapsto & (f_0,f_1)
\end{eqnarray*}
is an isomorphism of $\efq[X]$-modules.

Note that for all $r$, $M_r\subseteq M_0$ is a submodule, and let
$$
N_r =N_r(\synd{y},\alpha_1,\ldots,\alpha_{r},\beta_1,\ldots,\beta_{r})
:=\mu(M_r(\synd{y},\alpha_1,\ldots,\alpha_{r},\beta_1,\ldots,\beta_{r}))
$$
be the $\mu$-image of $M_r$ in $\efq[X]^2$. For obvious
reasons, we call $N_r$ the {\bf module of coefficient polynomials} of
$M_r$. Then by writing a typical element $(u,v)\in M_r$ as
$(u,v)=f_0(X)\bh_0 + f_1(X)\bh_1$ and substituting in the constraints
in the definition of $M_r$, we immediately obtain the following
characterization of $N_r$.

\begin{proposition}\label{prop:nr}
It holds that $N_r$ is the set of all pairs $(f_0,f_1)\in \efq[X]^2$
that satisfy the following condition:

\begin{itemize}

\item $\forall i\in\{1,\ldots,r\}$, 

\begin{enumerate}

\item $b_{0i} f_0(\alpha_i^{-1}) + b_{1i} f_1(\alpha_i^{-1})=0$

\item $b_{0i}f'_0(\alpha_i^{-1}) + c_{0i}f_0(\alpha_i^{-1})
+b_{1i}f'_1(\alpha_i^{-1}) + c_{1i}f_1(\alpha_i^{-1}) =0$,
where  
$$
b_{0i}:=h_{01}(\alpha_i^{-1}),\quad
b_{1i}:=h_{11}(\alpha_i^{-1}), 
$$
$$
c_{0i}:=h'_{01}(\alpha_i^{-1})+
\frac{\alpha_i}{\beta_ia_i}h_{00}(\alpha_i^{-1}), 
$$
$$
c_{1i}:=h'_{11}(\alpha_i^{-1})+
\frac{\alpha_i}{\beta_ia_i}h_{10}(\alpha_i^{-1}),
$$
(with $a_i:=\tl{a}_{i'}$ for the $i'$ with $\alpha_i=\lambda^{i'}$). 

\end{enumerate}

\end{itemize}

\end{proposition}
\begin{proof}
The conditions in the definition of $M_r$ (as a sub-module of $M_0$)
translate to the following conditions, both for all $i$: 
$$
(f_0h_{01}+f_1h_{11})(\alpha_i^{-1})=0, 
$$
and
$$
\beta_ia_i(f_0h_{01}+f_1h_{11})'(\alpha_i^{-1}) +
\alpha_i(f_0h_{00}+f_1h_{10})(\alpha_i^{-1}) =0.
$$
Dividing the second equation by $\beta_ia_i$, expanding, and
re-arranging terms, the proposition follows. 
\end{proof}

Note that it follows immediately from Theorem \ref{thm:main} that for
all $r$, $N_r$ is a module, since it is a homomorphic image of a
module. It also follows that the ``intermediate'' module, obtained by
intersecting $N_r$ only with the first constraint for $i=r+1$, is a module,
again, as a homomorphic image of a module. 

To use the minimality assertion of Theorem \ref{thm:main} also for the
coefficient-polynomial modules $N_r$, we have the following
proposition. 

\begin{proposition}
Let $w:=\deg(h_{11})-\deg(h_{00})-1$. Then for all
$(u_1,v_1),(u_2,v_2)\in M_0$, it holds that 
$$
\lmmo(u_1,v_1)<_{-1} \lmmo(u_2,v_2) \quad \iff \quad
\lm_w(\mu(u_1,v_1))<_{w}\lm_w(\mu(u_1,v_1)) 
$$
\end{proposition} 

\begin{proof}
For $(u,v)\in M_0$ write $(u,v)=f_0\bh_0 + f_1\bh_1$. With respect to
$<_{-1}$, the leading
monomial of $f_0\bh_0$ is $(X^{\deg(f_0)+\deg(h_{00})},0)$, while the
leading monomial of $f_1\bh_1$ is $(0,X^{\deg(f_1)+\deg(h_{11})})$
(we have used the fact that the leading monomial of $\bh_0$ is on the
left and the leading monomial of $\bh_1$ is on the right). Hence 
$$
\lmmo(u,v)= \begin{cases}
(X^{\deg(f_0)+\deg(h_{00})},0)& \text{if } \deg(f_0)+\deg(h_{00})\geq
\deg(f_1)+\deg(h_{11})\\
(0,X^{\deg(f_1)+\deg(h_{11})})& \text{if } \deg(f_0)+\deg(h_{00})\leq
\deg(f_1)+\deg(h_{11})-1  
\end{cases}
$$ 
It follows that
$$
\lmmo(u,v)= \begin{cases}
(X^{\deg(f_0)+\deg(h_{00})},0)& \text{if } \deg(f_0)\geq
\deg(f_1)+w + 1\\
(0,X^{\deg(f_1)+\deg(h_{11})})& \text{if } \deg(f_0)\leq
\deg(f_1)+w.  
\end{cases}
$$ 
We conclude that
\begin{equation}\label{eq:mua}
\lmmo(u,v)= (X^{\deg(f_0)+\deg(h_{00})},0) \quad\iff\quad 
\lm_w(f_0,f_1)=(X^{\deg(f_0)},0)
\end{equation}
and 
\begin{equation}\label{eq:mub}
\lmmo(u,v)= (0,X^{\deg(f_1)+\deg(h_{11})}) \quad\iff\quad 
\lm_w(f_0,f_1)=(0,X^{\deg(f_1)})
\end{equation}
and these are all the possible cases. The assertion now follows by
considering four possible cases of whether the $<_{-1}$-leading
monomial of $(u_1,v_1)$, $(u_2,v_2)$ is on the left/right, and 
(\ref{eq:mua}), (\ref{eq:mub}). 
\end{proof}

We therefore obtain the following corollary to Theorem
\ref{thm:main}:

\begin{corollary}
Write $(\omega,\sigma)=f_0\bh_0+f_1\bh_1$. If
$\dham(\bs{y},\bs{x})\leq t+r$, $\alpha_1,\ldots,\alpha_r$ are  
error locations and $\beta_1,\ldots,\beta_r$ are the corresponding
error values, then $(f_0,f_1)\in N_r$ and
$$
\lm_w(f_0,f_1)=\min\big\{\lm_w(g_0,g_1)|(g_0,g_1)\in
N_r\smallsetminus\{0\}\big\}. 
$$
\end{corollary}

The corollary shows that we can work directly with coefficient
polynomials, using \koetter{}'s iteration with respect to the monomial
ordering $<_{w}$. Also, the heuristic stopping criterion of
Appendix \ref{app:stop} works
just as before, because the discrepancies in $N_r$ are zero iff the
corresponding discrepancies in $M_r$ are zero.

The resulting application of \koetter{}'s iteration is listed below in
Algorithm C. Note that the coefficients $b_{0r},b_{1r},c_{0r},c_{1r}$
appearing in the calculation of $\Delta_j$ 
are defined in Proposition \ref{prop:nr}. Note
also that all required evaluations of $h_{00},h_{01},h_{10},h_{11}$
and their derivatives can be pre-computed once for all non-reliable
coordinates. 

\hfill\\ 

\begin{tabular}{|c|}
\hline
{\bf Algorithm C: \koetter{}'s iteration for coefficient vectors}\\
\hline
\end{tabular}

\begin{description}

\item[\bf{Input}] 

\begin{itemize}

\item A \groebner{} basis
$G=\{\bs{f}_0=(f_{00},f_{01}),\bs{f}_1=(f_{10},f_{11})\}$ 
for
$N_{r-1}(\synd{y},\alpha_1,\ldots,\alpha_{r-1},\beta_1,\ldots,\beta_{r-1})$,
with $\lm_w(\bs{f}_j)$ containing the $j$-th unit vector for 
$j\in\{0,1\}$, where $w:=\deg(h_{11})-\deg(h_{00})-1$

\item The next error location, $\alpha_r$, and the corresponding error
value, $\beta_r$

\end{itemize}

\item[\bf{Output}] A \groebner{} basis
$\Gplus=\{\bfplus_0=(\fplus_{00},\fplus_{01}),\bfplus_{1}=(\fplus_{10},\fplus_{11})\}$
for \\ $N_r(\synd{y},\alpha_1,\ldots,\alpha_r,\beta_1,\ldots,\beta_r)$
with $\lm_w(\bfplus_j)$ containing the $j$-th unit vector for $j\in\{0,1\}$ 

\item[\bf{Algorithm}]

\begin{itemize} 

\item For {\bf type} $=$ {\bf root}, {\bf der}

\begin{itemize}

\item If {\bf type} $=$ {\bf der}, 

\begin{itemize}

\item For $j = 0, 1$, set $\bs{f}_j:=\bfplus_j$ /* init:
output of {\bf root} iter.~*/

\end{itemize}

\item For $j=0, 1$, calculate 
$$
\Delta_j:=\begin{cases} b_{0r}f_{j0}(\alpha_r^{-1}) +
b_{1r}f_{j1}(\alpha_r^{-1})
 & \text{if {\bf type}$=${\bf root}}\\
 b_{0r}f'_{j0}(\alpha_r^{-1}) + c_{0r}f_{j0}(\alpha_r^{-1}) \\
 + b_{1r}f'_{j1}(\alpha_r^{-1}) + c_{1r}f_{j1}(\alpha_r^{-1}) & 
\text{if {\bf type}$=${\bf der}}
\end{cases}
$$ 
/* $b_{0r},b_{1r},c_{0r},c_{1r}$ defined in Prop.~\ref{prop:nr} */

\item Set $J:=\big\{j\in\{0,1\}|\Delta_j\neq 0\big\}$

\item For $j\in\{0,1\}\smallsetminus J$, set $\bfplus_j:=\bs{f}_j$

\item Let $j^*\in J$ be such that  $\lm_{w}(\bs{f}_{j^*})=\min_{j\in 
J}\{\lm_{w}(\bs{f}_j)\}$  

\item For $j\in J$ 

\begin{itemize}

\item If $j\neq j^*$

\begin{itemize}

\item Set $\bfplus_j := \bs{f}_j-\frac{\Delta_j}{\Delta_{j^*}}
\bs{f}_{j^*}$ 

\end{itemize}

\item Else /* $j=j^*$ */

\begin{itemize}

\item Set $\bfplus_{j^*} := (X-\alpha_r^{-1})\bs{f}_{j^*}$

\end{itemize}

\end{itemize}

\end{itemize}

\end{itemize}

\end{description}

\begin{remark}
{\rm
The validity of the update $\bfplus_{j^*} :=
(X-\alpha_r^{-1})\bs{f}_{j^*}$ for both the root and derivative
iterations can be proved as follows. First, it can
be verified directly that for both the root and derivative iterations,
$D(X\bs{f}_j)=\alpha_r^{-1}D(\bs{f}_j)$ (as done above for Algorithm A),
where $D$ is the linear functional of \koetter's iteration for the
respective iteration. A simpler way to prove this is as
follows. In the definition of $N_r$, we have implicitly defined
functionals (for the respective iterations) $D'\colon \efq[X]^2\to \efq$
by setting, for all $\bs{f}\in \efq[X]^2$,
$D'(\bs{f}):=D(\mu^{-1}(\bs{f}))$ for $D$ of Subsection \ref{sec:basic}. Hence  
\begin{multline*}
D'(X\bs{f})=D(\mu^{-1}(X\bs{f})) =
D(X\mu^{-1}(\bs{f}))=\alpha_r^{-1}D(\mu^{-1}(\bs{f}))=\alpha_r^{-1}D'(\bs{f}),
\end{multline*}
where the second equality follows from the $\efq[X]$-linearity of
$\mu^{-1}$, and the third equality follows from what we have already
proved for Algorithm A (where we assume, as before, that the derivative
iteration comes after the root iteration).
}
\end{remark}

The version of fast Chase decoding using Algorithm C is initiated on
the root of the tree $T$ with the \groebner{} basis $\{(1,0),(0,1)\}$ for
$N_0=\efq[X]^2$. When the heuristic stopping condition of Appendix
\ref{app:stop} holds, one can perform exhaustive substitution in one of
two ways, which we shall now describe. For short, in the following we
let $\bs{f}=(f_0,f_1)$ be the pair with the minimum $<_w$-leading
monomial from $\{\bfplus_0,\bfplus_1\}$ in the derivative iteration
for adjoining $\alpha_r$.  
\begin{enumerate}

\item Re-construct an estimated ELP (up to a non-zero multiplicative
constant) as $\hat{\sigma}(X)=f_0(X)h_{01}(X)+f_1(X)h_{11}(X)$ and
evaluate.  

\begin{itemize}
\item For the method for ruling out indirect hits of Appendix
\ref{app:stop}, we can then readily calculate the derivative
$\hat{\sigma}'(X)$ and evaluate it.
\end{itemize}

\item Calculate and store in advance the evaluations
$\{h_{01}(z^{-1})\}_{z\in \efq^*}$, $\{h_{11}(z^{-1})\}_{z\in
\efq^*}$. Now only the low-degree polynomials $f_0,f_1$ need to be
evaluated for calculating the evaluations
$\hat{\sigma}(z^{-1})=f_0(z^{-1})h_{01}(z^{-1}) +
f_1(z^{-1})h_{11}(z^{-1})$ for all $z\in \efq^*$.

\begin{itemize}
\item For the method for ruling out indirect hits of Appendix
\ref{app:stop},  we can also calculate and store in
advance the evaluations of derivatives
$\{h'_{01}(z^{-1})\}_{z\in \efq^*}$, $\{h'_{11}(z^{-1})\}_{z\in
\efq^*}$, and then calculate 
\begin{multline*}
\hat{\sigma}'(z^{-1})=f_0'(z^{-1})h_{01}(z^{-1})+f_0(z^{-1})h_{01}'(z^{-1})
\\ + f_1'(z^{-1})h_{11}(z^{-1})+f_1(z^{-1})h_{11}'(z^{-1}).
\end{multline*}
\end{itemize}

\end{enumerate} 

To bound the complexity of Algorithm C, we will need the following
proposition, in which we shall use a notation similar to that of
Definition \ref{def:itau} for Algorithm C instead of Algorithm A,
where ``$g$'' is replaced by ``$f$'' throughout. 

\begin{proposition}\label{prop:sumrdeg}
When Algorithm C is applied on an edge connecting a vertex at depth
$r-1$ to an edge at depth $r$ ($r\geq 1$), we have
\begin{equation}\label{eq:diag}
\deg(\fplus_{00}(r;\tau))+\deg(\fplus_{11}(r;\tau))\leq \begin{cases}
2r-1 & \text{if } \tau=\broot\\
2r   & \text{if } \tau=\bder,
\end{cases}
\end{equation}
and
\begin{equation}\label{eq:offd}
\max\{\deg(\fplus_{01}(r;\tau)),0\} +
\max\{\deg(\fplus_{10}(r;\tau)),0\}\leq \begin{cases} 
2r-2 & \text{if } \tau=\broot\\
2r-1   & \text{if } \tau=\bder.
\end{cases}
\end{equation}
\end{proposition}

\begin{remarknn}
{\rm
Note that the usage of $\max\{\deg(\cdot),0\}$ means that we sum only
over the degrees of non-zero polynomials.
}
\end{remarknn}

\begin{proof}
Recalling that the algorithm is initiated with the \groebner{} basis
$\{(1,0),(0,1)\}$ on the root of the decoding tree, (\ref{eq:diag})
follows by induction, as in each root and derivative iteration at most one
leading monomial is increased, and the increased leading monomial is
multiplied by $X$.

Since $\lm_{w}(\bfplus_0)$ is on the left and $\lm_{w}(\bfplus_1)$ is on
the right, we have 
$$
\deg(\fplus_{01}(r;\tau))<\deg(\fplus_{00}(r;\tau))-w
$$
and
$$
\deg(\fplus_{10}(r;\tau))\leq\deg(\fplus_{11}(r;\tau))+w.
$$
Summing the last two inequalities (and using (\ref{eq:diag})) proves
(\ref{eq:offd}) for the case 
where the two involved polynomials are non-zero. Also, if both
involved polynomials are zero, then there is nothing to prove. It
therefore remains to consider the case where one of the polynomials is
zero and the other is non-zero. 

For this case, we will prove that for all
$r\geq 1$, 
\begin{equation}\label{eq:oneoff}
\deg(\fplus_{01}(r;\tau))\leq\begin{cases}
2r-2 & \text{if }\tau=\broot \\
2r-1 & \text{if }\tau=\bder
\end{cases}
\end{equation}
(a similar proof works also for $\fplus_{10}(r;\tau)$).
For $r=1$, (\ref{eq:oneoff}) can be verified directly by checking $4$
options of $j^*$ in the root and derivative iterations.

Assume  by induction that $r\geq 2$ and that (\ref{eq:oneoff})
holds for $r-1$. For $\tau=\broot$, there are $3$ options to consider: If
$\Delta_0=0$ (no update), then (\ref{eq:oneoff}) obviously holds
for $\tau=\broot$ and $r$. If $\Delta_0\neq 0$ and $j^*=0$, then
$\deg(\fplus_{01}(r;\broot))=\deg(\fplus_{01}(r-1;\bder))+1$, and 
again (\ref{eq:oneoff}) holds for $\tau=\broot$ and $r$. Finally, if
$\Delta_0\neq 0$ and $j^*=1$, then considering the update rule in this
case and (\ref{eq:diag}) for $r-1$ and $\tau=\bder$, it follows again that
(\ref{eq:oneoff}) holds for $\tau=\broot$ and $r$. So, the induction
hypothesis implies (\ref{eq:oneoff}) for $r$ and
$\tau=\broot$. Applying the same arguments again, it can be shown that
(\ref{eq:oneoff}) holds also for $r$ and $\tau=\bder$. 
\end{proof}

Let us now proceed to bounding the complexity of Algorithm C on an
edge connecting a vertex at depth $r-1$ to a vertex at depth $r$ (for
$r\geq 1$).

\begin{itemize}

\item {\bf In the root iteration:} 

\begin{itemize} 

\item {\bf Evaluation:} Running over $j=0,1$ we eventually have to
evaluate once each of $\fplus_{00}(r-1;\bder)$,
$\fplus_{11}(r-1;\bder)$, $\fplus_{01}(r-1;\bder)$, and
$\fplus_{10}(r-1;\bder)$. Using Proposition 
\ref{prop:sumrdeg}, the required number of multiplications is at most
$2(r-1) + 2(r-1)-1 = 4r-5$. Also, for each of $j=0,1$, we have $2$ 
multiplications by a scalar (an overall of $4$ such multiplications), for
a total of $4r-1$ multiplications.

\item {\bf Multiplication of a polynomial by a constant:} For
$j\neq j^*$, we have to calculate the constant $\Delta_j/\Delta_{j^*}$,
which requires a single multiplication (assuming, as before, that we
have a table for calculating inverses), and to multiply $4$
polynomials whose sum of degrees is at most $4r-5$ by a constant
(recall Footnote \ref{footnote:jstar}). This
requires a total of $1 + 4r-5 +4=4r$ multiplications (the ``$+4$''
accounts for the fact that a polynomial of degree $d$ has $d+1$
coefficients).   

\end{itemize}

\item {\bf In the derivative iteration:}

\begin{itemize} 

\item {\bf Evaluation:} Running over $j=0,1$, we eventually have to
evaluate once each of $\fplus_{00}(r;\broot)$, $\fplus_{11}(r;\broot)$,
$\fplus_{01}(r;\broot)$, $\fplus_{10}(r;\broot)$, and their
derivatives. Taking the worst case assumption of characteristic $\neq
2$ and using Proposition \ref{prop:sumrdeg} again, this requires at
most 
$$
(2r-1)+(2r-2) + (2r-1-2)+(2r-2-2)=8r-10
$$ 
multiplications. There are also $4$ multiplications for calculating
$c_{0r}$, $c_{1r}$, and $8$ additional multiplications after
the substitutions, giving a total of at most $8r+2$ multiplications.

\item {\bf Multiplication of a polynomial by a constant:} Similarly to
the root iteration (but now with $r$ instead of $r-1$ and $\broot$
instead of $\bder$ in the bound of Proposition \ref{prop:sumrdeg}),
this gives a total of at most $1+(2r-1)+(2r-2)+4=4r+2$
multiplications. 
\end{itemize}

\end{itemize}

Summing all the above bounds, we obtain that the total number of
multiplications for moving from depth $r-1$ to depth $r$ with 
algorithm $C$ is at most
$$
M_C = 4r-1+ 4r + 8r+2 + 4r+2 =20r+3.
$$
Comparing this with $M_{\mathrm{Wu}}=10(t+r)+7$ calculated in
Subsection \ref{app:pairs}, we see that the complexity of Algorithm C 
is lower for each $r\leq \rmax$ when $\rmax<t$, since $2r<t+r$. Note that
as before, this complexity calculation does not account for the
(heuristically, rare) unrequired exhaustive root searches.

\section*{Acknowledgment}
We thank Avner Dor, Itzhak Tamo, and Moshe Twitto for their helpful
comments. We would also like to thank Johan Rosenkilde for his helpful
comments and for pointing us to \cite{BL94}, \cite{BL97}, \cite{NZ03},
and \cite{RS21}. Finally, we would like to thank the anonymous
reviewers for their helpful comments.


\end{document}